\documentclass{llncs}

\usepackage{times}

\usepackage{wrapfig}
\usepackage{booktabs}
\usepackage{subcaption}
\usepackage{alltt}
\usepackage{amsmath}
\usepackage{amssymb}
\usepackage{xspace}
\usepackage{stmaryrd}
\usepackage{graphicx}
\usepackage{color}
\usepackage{listings}
\usepackage{paralist}
\usepackage{arydshln}
\usepackage{semantic}
\newcommand{\cbstart}{}
\newcommand{\cbend}{}
\newcommand{\cbdelete}{}

\usepackage{hyperref}
\usepackage{cleveref}

\crefname{lemma}{lemma}{lemmas}
\Crefname{lemma}{Lem.}{Lem.}
\crefname{theorem}{theorem}{theorems}
\Crefname{theorem}{Thm.}{Thm.}
\Crefname{section}{Sec.}{Sec.}
\Crefname{definition}{Def.}{Def.}
\Crefname{figure}{Fig.}{Fig.}

\makeatletter
\newcommand{\xRightarrow}[2][]{\ext@arrow 0359\Rightarrowfill@{#1}{#2}}
\makeatother

\newcommand{\ignore}[1]{}

\newcommand{\para}[1]{\vspace{2pt}\noindent\textbf{\textit{#1}.}}

\newcommand{\fw}{\textit{fw}}
\newcommand{\type}{\texttt{type}}
\newcommand{\is}{\textit{is}}
\newcommand{\requested}{\texttt{requested}}
\newcommand{\trusted}{\texttt{trusted}}

\newcommand{\exref}[1]{Ex.~\ref{examp:#1}}

\newcommand{\secref}[1]{Sec.~\ref{sec:#1}}
\newcommand{\appref}[1]{App.~\ref{sec:#1}}
\newcommand{\lemref}[1]{Lem.~\ref{lem:#1}}
\newcommand{\defref}[1]{Def.~\ref{def:#1}}

\newcommand{\figref}[1]{Fig.~\ref{fig:#1}}
\newcommand{\tabref}[1]{Tab.~\ref{tab:#1}}
\newcommand{\thmref}[1]{Thm.~\ref{thm:#1}}

\newcommand{\AMDL}{\textrm{AMDL}}
\newcommand{\gV}[1]{\langle #1\rangle}

\newcommand{\bin}{\textbf{in}}
\newcommand{\band}{\textbf{and}}

\newcommand{\bnot}{\textbf{not}}

\newcommand{\N}{\mathbb{N}}

\newcommand{\pfld}{\textit{pfld}}
\newcommand{\const}{\textit{const}}

\newcommand{\prog}{\textit{prog}}
\newcommand{\paths}{\textit{paths}}
\newcommand{\cond}{\textit{cond}}
\newcommand{\atom}{\textit{atom}}
\newcommand{\action}{\textit{action}}
\newcommand{\gc}{\textit{gc}}
\newcommand{\pblock}{\textit{pblock}}
\newcommand{\True}{\textsf{True}}
\newcommand{\False}{\textsf{False}}
\newcommand{\lsyn}{|[}
\newcommand{\rsyn}{|]}
\newcommand{\Cont}{\textit{Cont}}

\newcommand{\tup}[1]{\langle #1 \rangle}
\newcommand{\rels}{\textit{rels}}

\newcommand{\Net}{\mathsf{N}}

\newcommand{\TrN}[2]{\xRightarrow{#1}_{\text{#2}}}

\newcommand{\TrM}[1]{\xrightarrow{#1}}
\newcommand{\TrMRel}[1]{\xrightarrow{#1}_{\text{R}}}
\newcommand{\TrMPack}[1]{\xrightarrow{#1}_{\text{P}}}
\newcommand{\TrMPackM}[1]{\xrightarrow{#1}_{\text{P},m}}
\newcommand{\TrMPackSub}[1]{\xrightarrow{#1}_{\text{P}[p,\tilde{p}]}}

\newcommand{\TODO}[1]{{\color{red} \textit{ \textbf{ TODO: #1 }}}}

\newcommand{\DONE}[2]{}

\newcommand{\eqdef}{\buildrel \mbox{\tiny\rm def} \over =}

\newcommand{\powerset}[1]{\pow{#1}}

\newcommand{\BB}[2][]{\llbracket #2 \rrbracket_{#1}}
\newcommand{\BBsub}[2]{\llbracket #1 \rrbracket_{#2}}

\newcommand{\pow}[1]{\wp(#1)}

\newcommand{\join}{\sqcup}

\newcommand{\Nat}{\mathbb{N}}
\newcommand{\lfp}{\textit{LeastFixpoint}}

\newcommand{\err}{\textit{err}}
\newcommand{\Tr}{\mathcal{T}}

\newcommand{\hosts}{H}
\newcommand{\packets}{P}
\newcommand{\packet}{p}
\newcommand{\mboxes}{M}

\newcommand{\mboxe}{m}

\newcommand{\ps}{\textit{ps}}
\newcommand{\rel}{\textit{rel}}
\newcommand{\atoms}{\textit{atoms}}
\newcommand{\RState}{\Sigma^{\text{R}}}
\newcommand{\PState}{\Sigma^{\text{P}}}
\newcommand{\Pack}{\text{P}}

\newcommand{\post}{\textit{post}}
\newcommand{\postLS}{\textit{postLS}}
\newcommand{\postLP}{\textit{postLP}} 

\newif\iflong  
\newif\iflonglong

\newif\ifexamples

\longtrue
\longlongfalse
\examplesfalse

\begin{document}
\frontmatter          
\pagestyle{headings}  
%
\mainmatter              
%
\title{Abstract Interpretation of Stateful Networks}
%
%


\author{Kalev Alpernas\inst{1}\and
Roman Manevich\inst{2}\and
Aurojit Panda\inst{3}\and
Mooly Sagiv\inst{1}\and
Scott Shenker\inst{4}\and
Sharon Shoham\inst{1}\and
Yaron Velner\inst{5}
}
\institute{Tel Aviv University\and
Ben-Gurion University of the Negev\and
NYU \and
UC Berkeley \and
Hebrew University of Jerusalem
}

\maketitle              

\begin{abstract}
Modern networks achieve robustness and scalability by maintaining
states on their nodes. These nodes are referred to as middleboxes
and are essential for network functionality.
However, the presence of middleboxes drastically complicates the
task of network verification.
Previous work showed that the problem is undecidable in general and EXPSPACE-complete
when abstracting away the order of packet arrival.


We describe a new algorithm for conservatively checking isolation properties
of stateful networks.
%
The asymptotic complexity of the algorithm is polynomial in the size of the network,
albeit being exponential in the maximal number of queries of the local state that a middlebox can do, which is often small.

Our algorithm is sound, i.e., it can never miss a violation of safety but may
fail to verify some properties.
The algorithm performs on-the fly abstract interpretation by (1)~abstracting away the order of packet processing and the number of times each packet arrives, (2)~abstracting away correlations between states of different middleboxes and channel contents, and (3)~representing middlebox states by their effect on each packet separately, rather than taking into account the entire state space. 
We show that the 
abstractions do not lose precision when middleboxes may reset in any state. This is encouraging since many real middleboxes reset, e.g., after some session timeout is reached or due to hardware failure.
\end{abstract}

\ignore{
\begin{abstract}
Modern networks achieve robustness and scalability by maintaining
states on their nodes. These nodes are referred to as middleboxes
and are essential for network functionality.
However, the presence of middleboxes drastically complicates the task of network verification.

We describe a new algorithm for conservatively checking the safety
of stateful networks. Our algorithm is modular in the sense that it
repeatedly analyzes each middlebox separately w.r.t. an intermediate
global state view.
%
Its asymptotic complexity is polynomial in the size of the network,
albeit being exponential in the maximal number of queries of the local state that a middlebox can do, which is often small.

Our algorithm is sound, i.e., it can never miss a violation of safety but may
fail to verify some properties.
The algorithm performs on-the fly abstract interpretation by (1)~abstracting away the order of packet processing, (2)~abstracting away correlations between states of different middleboxes and channel contents, and (3)~representing middlebox states by their effect on each packet separately, rather than taking into account the entire state space. 
We show that the 
abstractions do not lose precision when the middlebox may reset in any state. This is encouraging since many real middleboxes reset, e.g., after some session timeout is reached or due to hardware failure.
\end{abstract}
}

\section{Introduction}
\label{sec:Intro}

Modern computer networks are extremely complex, leading to many bugs and vulnerabilities that affect our daily life.
Therefore, network verification is an
increasingly important topic addressed by the programming languages
and networking communities~\cite{conext:uzniarPCVK12,nsdi:CaniniVPKR12,nsdi:KVM12,CCR:KhurshidZCG12,nsdi:KCZCMW13,FMACAD:SNM13,FlowLog,NetKat15}.
Previous network verification tools leverage a simple network
forwarding model, which renders the datapath {\em immutable}. That is,
normal packets going through the network do not change its
forwarding behaviour, and the control plane explicitly alters the
forwarding state at relatively slow time scales.

While the notion of an immutable datapath supported by an assemblage of routers
makes verification tractable, it does not reflect reality. {\em Middleboxes} are
widespread in modern enterprise networks~\cite{sherry2012making}. A simple
example of a middlebox is a stateful firewall which permits traffic from
untrusted hosts only after they have received a packet from a trusted host.
Middleboxes, such as firewalls, WAN optimizers, transcoders, proxies, load-balancers
and the like, are the most common way to insert new functionality in
the network datapath, and are commonly used to improve network performance and
security. Middleboxes maintain a state and may change their state and forwarding
behavior in response to packet arrivals. While useful, middleboxes are a common
source of errors in the network~\cite{potharaju2013demystifying}.


\begin{figure}[t]
    \centering
    \includegraphics[width=0.7\textwidth]{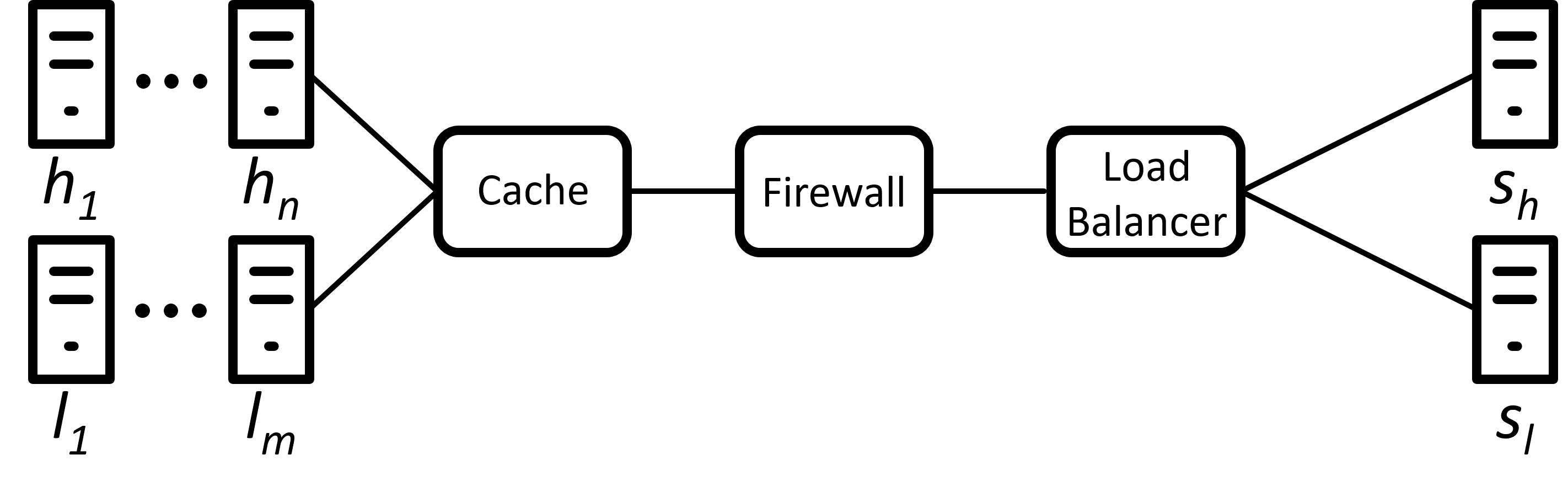}
    \caption{A middlebox chain with a buggy topology.}
    \label{fig:example}
\end{figure}


As a simple example, consider the middlebox chain described in \figref{example}. In this network, a firewall is used to ensure that low security hosts ($l_1,\ldots,l_m$) do not receive packets from the $S_h$ server, and a cache and load balancer are used to improve performance.
Unfortunately, the configuration of the network is incorrect since the cache may respond with a stored packet, bypassing the security policy enforced by the firewall.
Swapping the order of the cache and the firewall results in a correct configuration.

\para{Safety of Stateful Networks}
We address the problem of verifying safety of networks with middleboxes, referred to as \emph{stateful networks}.
We target verification of \cbstart \emph{isolation} properties, namely, that packets sent from one host (or class of hosts) can never reach another host (or class of hosts). \cbend
Yet, our approach is sound for any safety property.
%
%
For example, it detects the safety violation described in \figref{example}, and verifies the safety of the correct configuration of this network.

Our focus is on verifying the configuration of stateful networks, i.e.,
addressing errors that arise from the interactions between middleboxes, and not
from the complexity of individual middleboxes. Hence, we
follow~\cite{velner2016some} and use an abstraction of middleboxes as finite-state programs. Previous
work~\cite{velner2016some,DBLP:conf/sigcomm/SivaramanCBKABV16} has shown that
many kinds of middleboxes, including proxy, cache proxy, NAT, and various kinds
of load-balancers can be modeled in this way, sometimes using non-determinism to
over-approximate the behaviour, e.g. to model timers, counters, etc. Since we
are interested in safety properties, such an abstraction (overapproximation) is
suitable.

As shown in~\cite{velner2016some}, it is undecidable to check safety properties in general and isolation in particular, even for
middleboxes with a finite state space, and even when the order of packets
pending for each middlebox is abstracted away the complexity is quite high
(EXPSPACE-complete). Therefore, in this paper we develop additional abstractions
for scaling up the verification. 

\para{Our approach}
This paper makes a first attempt to apply abstract
interpretation~\cite{cousot1979systematic} to automatically prove the safety
of stateful networks.
Our approach combines sound network-level abstractions and middlebox-level abstractions that, together, make the verification task tractable.
Roughly speaking, we apply
\begin{inparaenum}[(i)]
\item order abstraction~\cite{velner2016some}, abstracting away the order of packets on channels,
\item counter abstraction~\cite{pnueli2002liveness}, abstracting away their cardinality,
\item network-level Cartesian abstraction~\cite{cousot1979systematic,DBLP:journals/tcs/FlanaganFQS05,POPL:JMP17}, abstracting away the correlation between the states of different middleboxes and different channel contents, and
\item middlebox-level Cartesian abstraction, abstracting away the correlation between states of different packets within each middlebox.
\end{inparaenum}

%
%


The network-level abstractions, (i)-(iii),
lead to a chaotic iteration algorithm that is polynomial in the state space of the individual middleboxes and packets.
However, the number of middlebox states can be exponential in the size of the network.
For example, a firewall may record the set of trusted hosts and thus its states are subsets of hosts.
Therefore, the 
resulting analysis is exponential in the number of hosts\footnote{Unfortunately, if the set of hosts is not fixed, the safety problem becomes undecidable (even under the unordered abstraction) (\Cref{sec:small-model}). This
means that, in general, it is not possible to alleviate the dependency of the complexity on
the hosts.}.
%
%

The middlebox-level Cartesian abstraction, (iv), 
is the key to reducing the complexity to polynomial.
The crux of this abstraction is the observation that the abstraction of middleboxes as reactive processes that query and update their state in a restricted way (e.g.,~\cite{velner2016some}) 
allows to represent a middlebox state as a product
of loosely-coupled \emph{packet states}, one per potential packet. This lets us
define a novel, non-standard, semantics of middlebox programs that we call
\emph{packet effect semantics}. The packet effect semantics is equivalent
(bisimilar) to the natural semantics. However, while the natural semantics is
monolithic, the packet effect semantics decomposes a single middlebox state into
the parts that determine the forwarding behavior of different packets, and
therefore facilitates the use of Cartesian abstraction to further reduce the
complexity.




One of the main challenges for abstract interpretation is evaluating its
precision. To address this challenge, we provide sufficient conditions that
ensure precision of our analysis. Namely, we show that if the network is safe in
the presence of packet reordering and middlebox reverts, where a middelbox may
revert to its initial state at any moment, then our analysis is guaranteed to be
precise, and will never report false alarms.
This is, to a great extent, due to the packet effect semantics, which allows to
use a middlebox-level Cartesian abstraction without incurring additional
precision loss for such networks.
%
Notice that middlebox reverts enable modelling arbitrary hardware failures,
which have not been addressed by previous work on stateful network verification
(e.g., in \cite{velner2016some}).
Surprisingly, verification becomes easier under the assumption that middleboxes
may reset at any time. (Recall that for arbitrary unordered networks safety
checking is EXPSPACE-complete.)

In summary, the main contributions of this paper are
\begin{itemize}
\item We introduce the first abstract interpretation algorithm for verifying
safety of stateful networks, whose time complexity is polynomial in the size of
the network, albeit exponential in the maximal number of queries of the local
state that a middlebox can do, which is often small even for complex middelboxes
(up to 5 in our examples).
\item We develop \emph{packet effect semantics}, a non-standard semantics of
middelbox programs that facilitates middlebox-level Cartesian abstraction,
reducing the complexity of the abstract interpretation algorithm from
exponential in the size of the network to polynomial without incurring any
additional precision loss for unordered reverting networks.
\item We provide sufficient conditions for precision of the analysis that have a
natural interpretation in the domain of stateful networks: ignoring the order of
packet processing and letting middleboxes revert to their initial states at any
time.
\item We prove lower bounds on the complexity of safety verification in the
presence of packet reordering and/or middlebox reverts, showing that our
algorithm is essentially optimal.
\item We implement our analysis and show that it scales well with the number of
hosts and middelboxes in the network.
\end{itemize}



\ignore{
The key idea is to use Cartesian abstract interpretation~\cite{cousot1979systematic,DBLP:journals/tcs/FlanaganFQS05,POPL:JMP17}, which ignores correlations between different parts of a system.

Cartesian abstraction is natural to use in systems where the semantics ranges over loosely-coupled components.
This makes it straightforward to apply on computer networks, where each component (e.g., middlebox) acts as an individual unit and cannot directly access the internal state of other units.
On the other hand, for middleboxes, Cartesian abstraction may not seem like the best fit due to their sequential semantics. However, we observe
that the middlebox behavior for many real middleboxes can be expressed in a simple way (captured by a modeling language that we define) that allows to represent a middlebox state as a product of loosely-coupled \emph{packet states}, one per potential packet.
This representation facilitates the use of Cartesian abstraction to reduce the state space of the individual middleboxes.


The combination of the abstractions makes verification tractable.
Interestingly, we show that the abstractions we employ are as precise as ignoring the order of packet processing and letting middleboxes revert to their initial states at any time.
}


\ignore{
\subsection{Abstractions for Safety of Networks of Middleboxes}

\para{Network-Level Abstractions}
At the network level, our analysis combines the following, rather standard, abstractions:
\begin{inparaenum}[(i)]
\item Cartesian abstraction~\cite{cousot1979systematic,DBLP:journals/tcs/FlanaganFQS05,POPL:JMP17} that abstracts away correlations between the local states of
individual middleboxes and the contents of the individual channels,
\item  Unordered channels abstraction~\cite{velner2016some} that abstracts away the order of packets on channels, and
\item Counter abstraction~\cite{pnueli2002liveness} that abstracts away the number of copies of packets on channels.
\end{inparaenum}

\ignore{
At the network level, we combine Cartesian abstraction with unordered channels abstraction~\cite{velner2016some} and a counter abstraction~\cite{pnueli2002liveness} for the occurrence of packets on channels.
Namely, our analysis abstracts away correlations between the local states of
individual middleboxes and the contents of the individual channels.
This means that we verify the correctness of each middlebox against all possible reachable states of other middleboxes and channels.
In addition, the contents of channels are abstractly represented
by tracking for each channel the set of packets that may reside on it. In particular, this means that the order of packets on channels is abstracted away (as in~\cite{velner2016some}), but also their multiplicity (i.e., how many copies of the packet appear in the channel), and the correlations between different packets.
}

These abstractions lead to a chaotic iteration algorithm that is polynomial in the state space of the individual middleboxes and packets.
However, the number of middlebox states can be exponential in the size of the network.
For example, a firewall may record the set of trusted hosts and thus its states are subsets of hosts.
Therefore, the 
resulting analysis is exponential in the number of hosts and thus will not scale.

Unfortunately, if the set of hosts is not fixed, the safety problem becomes
undecidable (even under the unordered abstraction) (\Cref{sec:small-model}). This
means that it is not possible to alleviate the dependency of the complexity on
the hosts.


\para{Middelbox-Level Cartesian Abstraction}
Our next step is therefore to tackle the state-space of individual middleboxes by applying middlebox-level abstraction.
The crux of the middlebox-level abstraction is the formulation of a non-standard semantics of middlebox programs that we call \emph{packet effect semantics}.
The packet effect semantics is equivalent to the natural semantics. However, while the natural semantics is monolithic, the packet effect semantics decomposes a single middlebox state into the parts that determine the forwarding behavior of different packets, and therefore facilitates the use of Cartesian abstraction.
In order to explain the packet effect semantics, we first need to provide some more details on the way we model middleboxes.

\ignore{
Our next step is therefore to tackle the state-space of individual middleboxes by applying middlebox-level Cartesian abstraction.
This step is not straightforward due to the sequential semantics of middleboxes. It is facilitated by an equivalent semantics of middlebox programs
that decomposes a single middlebox state into the parts that determine the forwarding behavior of different packets.
Roughly speaking, we abstract away correlations between these parts.
%
This means that we consider every forwarding behavior that is feasible for a packet at some middlebox state as feasible in all reachable states.
In order to explain the abstraction more precisely, we first need to provide some more details on the way we model middleboxes.
}


\para{Abstract Modeling Language}
We observe that the behavior of many real middleboxes can be expressed in a
simple way. Accordingly, we introduce the \textbf{A}bstract \textbf{M}iddlebox
\textbf{D}efinition \textbf{L}anguge (\AMDL) for defining the effect of
middleboxes.
The language is similar to \cbstart OCCAM~\cite{roscoe1988laws} and CSP~\cite{CSP:Hoare:1978}\cbend.
Each Middlebox maintains a state using uninterpreted relations over finite domains and uses restricted state queries and relation updates to define how a packet is processed and its effect on the relations and the output packet(s).
In this way, AMDL provides a concise way to define finite state transducers. 

\ignore{
AMDL restricts the ability of the middlebox to generate new packets.
The code can only rewrite packets using values of the current processed packet or configuration constants.
For example, swapping source and destination is allowed but adding one to destination is forbidden.
The effect of a packet on the state is restricted in a similar way. 
}

Despite its simplicity, AMDL allows modeling many realistic middleboxes
while hiding implementation details by using non-determinism.
In this work we do not address the question of proving that the actual (usually complex and proprietary) middlebox code implements the
required effect.
\ignore{
Our language also permits defining safety properties such as isolation via middleboxes that may abort.
}

\para{Summarizing Packet Effects}
\ignore{
A naive way to apply Cartesian abstraction on an AMDL program would be to ignore correlations between different relations.
However, such an abstraction would be too coarse. Furthermore, since the state space of a single relation may already be exponential, it would not mitigate the problem.
Instead, 
we realize the middlebox-level Cartesian abstraction on a different representation of middlebox states.
}%
The key idea in our middlebox-level abstraction is the observation that 
we can view a single middlebox state as a mapping from packets to their \emph{packet states}, where the packet state of a packet $p$ suffices to determine both the sequence of packets that the middlebox outputs when it receives $p$ as an input packet, and the way the state changes in response to $p$.
Technically, the packet state of packet $p$ records the valuation of the Boolean conditions in the code when evaluated on $p$. In fact, it suffices to record the valuations of \emph{state queries} --- conditions that examine the middlebox relations and capture the dependency of the middlebox behavior on its state.
For example, a state query could check whether the source of the current packet was seen in the past.
Notably, the mapping of packets to packet states defines a bisimilar representation of the middlebox states.

\ignore{
we first identify a bisimilar representation of middlebox states,
which records the state queries in the \AMDL\ code that match a given input packet.
The idea is that middlebox behavior written in AMDL is local in the following way: we can view a single middlebox state as a mapping from packets to their \emph{packet states}, where the packet state of $p$ suffices to determine both the sequence of packets that the middlebox outputs when it receives $p$ as an input packet, and the way the state changes in response to $p$.
Technically, the packet state of packet $p$ records the valuation of the Boolean conditions in the code when evaluated on $p$. In fact, it suffices to record the valuations of \emph{state queries} --- conditions that examine the middelbox relations and capture the dependency of the middlebox behavior on its state.
For example, a state query could check whether the source of the current packet was seen in the past.
Formally, the mapping of packets to packet states defines a bisimilar representation of the middlebox states.
}

%

\para{Effectively Polynomial Abstract Interpretation Algorithm}
Applying a middlebox-level Cartesian abstraction on the packet-effect
representation of middleboxes results in an abstract interpretation algorithm
that enumerates, for each middlebox, the 
potential packet states of individual packets separately (based on their
headers), instead of enumerating (entire) middlebox reachable states as
valuations of relations.

The size of the packet state space is quadratic in the number of hosts, and
exponential in the number of distinct state queries.
The latter is typically a small constant even for complex middelboxes. In our
examples, the maximal number of state queries is 5.

Therefore, the abstract interpretation algorithm is realized with the packet-effect representation in time that is polynomial in the size of the network and
packet-header space, albeit exponential in the maximal number of
distinct state queries in the AMDL code of an individual middlebox.

\subsection{Precision and Complexity Bounds} 
One of the biggest hurdles in adopting abstract interpretation is false alarms.
Therefore, we identified a property that is applicable to many networks
and which guarantees the absence of false alarms.

\para{Reverting Unordered Networks}
A middlebox is \emph{reverting} if from every state it can reset to its initial state.
Notice that resets are similar to soft network states~\cite{DBLP:conf/icnp/LuiMR04}, and could be triggered, e.g., by session timeouts or hardware failures.

We show that our Cartesian abstractions do not incur a loss of precision on networks in which middleboxes are reverting and the order of packet arrivals is ignored.
Therefore, for such networks our algorithm is a decision procedure. 
Surprisingly, the verification becomes easier under the assumption that middleboxes may reset at any point in time. (Recall that for arbitrary unordered networks the safety problem is EXPSPACE-complete.)
For other networks our algorithm is sound but may produce false alarms if the network is in fact safe but safety is violated under a reverting unordered semantics.


\para{Lower-Bounds}
We show that it is undecidable to check the safety of reverting networks when packet order is respected, i.e., channels have a FIFO behavior.
\iflonglong
This indicates that an attempt to design a more precise abstract interpretation algorithm which is precise for reverting but ordered networks is doomed.

\else
This justifies the need for the unordered abstraction even in such networks.
\fi
We also show that the complexity of checking safety in reverting unordered networks is coNP-hard. 
\iflonglong
Since our algorithm is a decision procedure for this problem we conclude that we cannot strive for an abstract interpretation algorithm that has the same precision guarantees
but avoids the exponential dependency of the algorithm on the number of state queries.
\else
In this sense, our algorithm, which is a decision procedure for this problem, is optimal.
\fi


\iflonglong
\para{Limitations}
While our approach enables scalable verification for a variety of networks, it has a few limitations.
First, our verification time grows exponentially with the number of state queries. This means that the running time grows exponentially
with the amount of state that a middlebox must access. In practice, middleboxes access a limited amount of state per
packet so as to minimize their working-set size. Hence, this does not significantly limit our applicability.
Second, we rely on abstract packet models; this limits the accuracy with which we can model middleboxes, e.g.,
intrusion detection systems that search for byte patterns in the packet payload, etc.
%

This work focuses on proving isolation properties in stateful networks. For technical reasons, we model safety properties by middleboxes. This allows us to define a richer set of safety properties than just isolation properties. However, property middleboxes are also subject to the same abstractions. When the property middlebox itself maintains a state, the abstraction may add to it aborting executions and hence introduce false alarms. This may happen for history-sensitive properties such as ``a message with property $x$
can only be sent after a message with property $y$ has been sent'', cannot always be proved with our method.
Importantly, for network isolation, as well as stateful safety properties that we call \emph{revert-robust} (\secref{RevertRobust}), the abstraction of the property middelbox will not introduce false alarms.
\else

\para{Beyond Isolation}
This work focuses on proving isolation properties in stateful networks. For technical reasons, we model safety properties by middleboxes that may abort. This allows us to define a richer set of safety properties than just isolation properties. However, property middleboxes are also subject to the same abstractions. When the property middlebox itself maintains a state, the abstraction may add to it aborting executions and hence introduce false alarms.
%
Importantly, for network isolation, as well as stateful safety properties where the set of ``safe'' behaviors is suffix-closed (see \Cref{sec:RevertRobust}), the abstraction of the property middelbox will not introduce false alarms.
\fi


\subsection{Empirical Evaluation}
We have implemented our algorithm and applied it to several interesting topologies (\secref{empirical}).
Our implementation was able to detect, for example, the safety violation described in \figref{example} (detecting violations is guaranteed by the soundness of our analysis). It was also able to verify that the correct configuration of this network is safe (even when unbounded
number of packets are sent).


We evaluate the scalability of our approach under two scenarios:
(a) we vary the number of middleboxes in the network;
and (b) we vary the number of hosts connected to a network.
Our algorithm verified networks with thousands of hosts, and with hundreds of middleboxes in timescales of hours, suggesting that our
approach can be used to verify safety in real-world examples. \DONE{[3] KA: Maybe add specific run result numbers. Also, rephrase.}{RM: rephrased. Decided on no numbers.}

}

We defer proofs of key claims to  \iflong\appref{proofs} \else the supplementary materials\fi.

\section{Expressing Middlebox Effects}
\label{sec:amdl}

This section defines our programming language for modeling the
abstract behavior of middleboxes in the network. Our modeling language is independent of the particular network topology, which is defined in \secref{network-semantics}.
The proposed language, \AMDL\  (\textbf{A}bstract \textbf{M}iddlebox \textbf{D}efinition \textbf{L}anguage),
is a restricted form of OCCAM~\cite{roscoe1988laws},
%
similar to the languages of~\cite{velner2016some,DBLP:conf/sigcomm/SivaramanCBKABV16}.

We first define the syntax and informal semantics of \AMDL\ (\secref{syntax});
we then define a formal ``standard'' \emph{relation effect semantics} (\secref{relation-sem});
we continue by defining an alternative \emph{packet effect} semantics (\secref{packet-sem}),
which is bisimilar to the relation effect semantics (\secref{bisim}); and finally
we present a localized version of the packet effect semantics (\secref{local-sem}),
which is suitable for Cartesian abstraction.


\ignore{
Middleboxes are implemented as reactive processes.
Middleboxes can have states realized using relations; they read a packet and consequently produce
and send new packets, and potentially change the state of the process.
For example, the code for a session firewall is depicted in \figref{session-fw-code}.
}
%

\para{Packets}
Middlebox behavior in our model is defined with respect to packets that consist of a fixed, finite, number of packet fields, ranging over finite domains. As such, a packet $\packet \in \packets$ in our formalism is a tuple of packet fields over predefined finite sorts.
In our examples, a packet is a tuple $\tup{s,d,t}$,
where $s,d$ are the source and destination hosts, respectively, taken from a finite set of hosts $H$, 
and $t$ is a packet tag (or type) that ranges over a finite domain $T$. 
In this case, $|\packets|$ is polynomial in $|\hosts|$.
(Our approach is also applicable when additional fields are added, e.g., for modeling the packet's payload via an abstract finite domain.)

\subsection{Syntax and Informal Semantics}
\label{sec:syntax}

\figref{amdl-syntax} describes the syntax of the \AMDL\ language\footnote{In the
code examples, we write \texttt{p} for the triple \texttt{(src,dst,type)} and
use access path notation to refer to the fields, e.g., \texttt{p.src}.}.
Middleboxes are implemented as reactive processes, with events triggered by the
arrival of packets. If multiple packets are pending, the \AMDL\ process 
non-deterministically reads a packet from one of the incoming channels of the
process. The packet processing code is a loop-free block of guarded-commands,
which may update relations and 
forward potentially modified packets to some of the output ports.
\AMDL\ uses \emph{relations} over finite domains to store the middlebox
state. 
These are the only data structures allowed in \AMDL.
%
%
The only relation operations allowed are inserting a value to a relation, removing a value from a relation, and
\emph{membership queries} --- checking whether a value is in a relation.
For a membership query of the form $\overline{a} ~ \bin ~ \text{r}$, we denote the relation, r, used in the query by $\rel(q)$ and denote the tuple of atoms $\overline{a}$ by $\atoms(q)$.
For example, the code for a session firewall is depicted in \figref{session-fw-code}.
%

\begin{figure}[t]
\centering
\begin{alltt}
\begin{tabbing}
sf\=irewall = do \+ \\
in\=ternal_port ? p => \+ \\
if\= \+\\
p.dst in trusted => external_port ! p\\
\(\Box\) \\
p.\=type = 0 => // request packet \+ \\
external_port ! p; \\
requested(p.dst) := true \-\-\\
fi \-\\
\(\Box\)
\= \+ \\
ex\=ternal_port ? p => \+ \\
if\= \+ \\
p.src in trusted => internal_port ! p\\
 \(\Box\)\\
p.\=type = 1 and p.src in requested => \+\\
// response packet with a request\\
trusted(p.src) := true \- \- \\
fi \-\- \\
od
 \end{tabbing}
 \end{alltt}
  \caption{\label{fig:session-fw-code}%
  \AMDL\ code for 
    session firewall.}
\end{figure}

Middleboxes 
may enforce safety properties using the \textbf{abort} command.
For example, an isolation middlebox would abort when a forbidden packet is received.

\ignore{
The semantics of each middlebox consists of non-deterministic actions that
react to packet arrivals. 
The middlebox defines updates to its state and forwards potentially modified packets to some of its output ports.
%
When several packets can be read, one of them is chosen non-deterministically. If no packets are read, the process blocks until a packet arrives.
}







\begin{figure}
\[
\begin{array}{lcl}
\gV{\textit{mbox}} & ::= & m = \textbf{do}~ \gV{\pblock}~ [\Box~ \gV{\pblock}]^{*}~ \textbf{od}\\
\gV{\pblock} &::= &  c ~ \textbf{?}~ \overline{\pfld} ~ \Rightarrow~ \gV{\gc}\\
\gV{\gc} &::=  & \gV{\cond} \Rightarrow \gV{\action} 
 ~ | ~\textbf{if}~ \gV{\gc}~ [\Box~ \gV{\gc}]^{*}~ \textbf{fi}\\
 \gV{\action} &::= &
\gV{\action}~ \textbf{;}~ \gV{\action}  ~ | ~

c ~ !~ \overline{\gV{\atom}} ~ | ~
 r(\overline{\gV{\atom}})~ \textbf{:=} ~\gV{\cond} 
  ~ | ~ \textbf{abort}
  \\
 \gV{\cond} & ::=  &\textbf{true} ~ | ~
        \gV{\cond} ~\band ~ \gV{\cond} ~ | ~
         \bnot~ \gV{\cond} 
        ~ | ~ \gV{\atom} = \gV{\atom} ~ | ~  \overline{\gV{\atom}}~ \bin ~ r\\
\gV{\atom} & ::=  & \pfld ~|~ \const\\
\end{array}
\]
\caption{\label{fig:amdl-syntax}%
\AMDL\ syntax.
$\overline{e}$ denotes a comma-separated list of elements drawn from the domain
$e$. $\textbf{abort}$ imposes a safety condition. $c\ ?\ p$ reads $p$ from a
channel $c$ and $c\ !\ p$ writes $p$ into $c$. We write $m$ for a middlebox
name, $r$ for a relation name, and $c$ for a channel name.
We write $\const$ for a constant symbol and $\pfld$ for identifiers used to
match fields in packets, 
e.g., \texttt{src}.
Non-deterministic choice is denoted by $\Box$.
}
\end{figure}

\subsection{Middlebox Relation Effect Semantics} \label{sec:relation-sem}

We now sketch the semantics of \AMDL.
%
The definitions below supply a part of the
full network semantics, which is given in \secref{network-semantics}.


\para{Middlebox States}
%
Each middlebox $\mboxe\in \mboxes$ maintains its own local state as a set of
relations.
The domain of a relation $r$ defined over sorts $s_{1..k}$ is $D(r)\eqdef
D(s_1)\times\ldots\times D(s_k)$, where $D(s_i)$ is the domain of sort $s_i$.
%
%
%
We use $\rels(\mboxe)$ to denote the set of
relations in $\mboxe$, and $D(m)$ to denote the union of $D(r)$ over $r \in
\rels(\mboxe)$.
%

The \emph{middlebox state} of $\mboxe$ is then a function
$s \in \RState[\mboxe] \eqdef \rels(\mboxe) \to \powerset{D(m)}$,
mapping each $r \in \rels(\mboxe)$ to $v \subseteq D(r)$.
In addition, we introduce a unique \emph{error} middlebox state, denoted $\err$.
We assume that $\err \in \RState[\mboxe]$ for every middlebox $\mboxe$.

\para{Middlebox Transitions}
Middlebox transitions have the form $$\TrMRel{(p, c)/(p_i, c_i)_{i=1..k}}
\subseteq \RState[m] \times \RState[m]$$ where $(p,c)$ denotes packet-channel at
the input, and $(p_i, c_i)_{i=1..k}$ is the sequence of packet-channel pairs
that the middlebox outputs.

For example, for $s \eqdef [\requested \mapsto\emptyset, \trusted\mapsto\emptyset]$,
the guarded command corresponding to the internal port of the firewall middlebox (\figref{session-fw-code}) induces a transition
$s \TrMRel{((h_1,h_2,0),\overset\rightarrow{c_{in}})/((h_1,h_2,0),\overset\rightarrow{c_{out}})} s'$ where
$s' \eqdef [\requested\mapsto\{h_2\}, \trusted\mapsto\emptyset]$.

\textbf{abort} commands induce transitions to the $\err$ state.

%
The formal definition of the middlebox transitions appears in
\iflong \appref{amdl-semantics}\else the supplementary materials\fi.

\ignore{
\else
\begin{example}
\label{examp:rel-transition}
The first action in \tabref{CartesianExplicit} is due to the transition $s \TrMRel{((h_1,h_2,0),\overset\rightarrow{e_1})/((h_1,h_2,0),\overset\rightarrow{e_2})} s'$ of $\fw_1$ where the input and output states are defined as follows:
\[
\begin{array}{l}
s \eqdef [\requested\mapsto\emptyset, \trusted\mapsto\emptyset] \\
s' \eqdef $[$\requested\mapsto\{h_2\}, \trusted\mapsto\emptyset$]$ \enspace.
\end{array}
\]
The transition is achieved by a sequence of internal transitions which:
(1) choose the guarded command corresponding to the input port $\overset\rightarrow{e_2}$,
(2) read the input packet $(h_1,h_2,0)$ into the variable \texttt{p},
(3) choose the second guarded command,
(4) evaluate the expression ``\texttt{p.type=0}'' to \True,
(5) output the packet $(h_1,h_2,0)$ to the output port $\overset\rightarrow{e_3}$, and
(6) update the \requested\ relation by adding $h_2$.
\qed
\end{example}
\cbend
}




\ignore{
\para{Middlebox Transitions}
Formally, a middlebox $m$ defines a transition relation $\TrM{(p, c)/(p_i, c_i)_{i=1..k}} \subseteq \RState[m] \times \RState[m]$
where the packet $p$ is read from channel $c$ when the middlebox is in the source state and sends
packets $p_i$ on channels $c_i$ moving to the target state.
%
%
}

\subsection{Middlebox Packet Effect Semantics} \label{sec:packet-sem}

\cbstart
We now present a semantics that is equivalent to the relation effect semantics.
The semantics is based on an alternative (yet isomorphic) representation of middlebox states that
reveals a loose coupling between the parts of the state that are relevant for different packets.
This loose coupling then facilitates a Cartesian abstraction that abstracts away correlations between packets in the same state.
\cbend

\subsubsection{Packet Effect Representation of Middlebox State}

Recall that in \secref{syntax} we restrict the values that can be used in a
middlebox program to either constants or the values of fields of the currently
processed packet. We do not allow extracting tuples from the relation (e.g., by
having a \texttt{get} command, or by iterating over the contents of the
relation). Instead, we limit the interaction with the relation to checking
whether a tuple (that consists of packet fields or constants) exists in the
relation. Consequently, instead of storing the contents of all relations, the
state of the middlebox can be represented by mapping all potential packets in
the network to their effect on the middlebox. Specifically, we map each packet
and membership query in the program to whether that membership query will be
evaluated to $\True$ when the program is executed on that packet.



For every middlebox $m$, we denote by $Q(m)$ the set 
of membership queries in $m$'s program.
(We need not distinguish between different instances of the same query.)
For example, in \figref{session-fw-code}, $Q(\fw) = \{\texttt{p.dst in trusted},\
\texttt{p.src in trusted},\\
\texttt{p.src in requested}\}$.
%
\ignore{
Each membership query $q \in Q(m)$ is of the form $\overline{a} ~ \bin ~ \text{r}$. \cbdelete We denote the relation, r, used in a query by $\rel(q)$ and denote the tuple of atoms $\overline{a}$ by $\atoms(q)$.
Given a packet $p \in \packets$, we denote by $\atoms(q)(p)$ the result of substituting each field name in $\atoms(q)$ by its value in $p$.
Namely, $\atoms(q)(p) \in D(\rel(q))$.

\cbstart
\begin{example}
In \figref{session-fw-code}, we have the following:
\[
\begin{array}{l}
Q(\fw_1)= \{\texttt{p.dst in trusted},\
\texttt{p.src in trusted},\
\texttt{p.src in requested}\}\\
\begin{array}{rcl}
\rel(\texttt{p.dst in trusted})&=&\trusted\\
\atoms(\texttt{p.dst in trusted})&=&\texttt{dst}
\end{array}
\end{array}
\]
Finally, $\atoms(\texttt{p.dst in trusted})((h_2,h_1,0))=h_1$.
\qed
\end{example}
\cbend
}

\cbstart
The \emph{packet effect state} of a middlebox $m$ is a function $s \in \PState[m] \eqdef \packets \to Q(m) \to \{\True,\False\}$,
mapping each packet $\packet \in \packets$ to the 
evaluation of all queries of $m$ when $p$ is the input packet, thus capturing the way in which $\packet$ traverses $m$'s program.
We refer to $s(p) \in Q(m) \to \{\True,\False\}$ as the \emph{packet state} of packet $p$ in middlebox state $s$.
We extend $\PState[m]$ with an error state $\lambda p \in \packets. \ \err$, which is also denoted $\err$.



\subsubsection{Middlebox Transition Relation in the Packet Space}
\label{sec:packet-space-concrete-transformer}

The semantics of middlebox $m$ in the packet space is defined via a transition relation $\TrMPackM{(p, c)/(p_i, c_i)_{i=1..k}} \subseteq \PState[m] \times \PState[m]$.
When $m$ is clear, we omit it from the notation.
A transition $\tilde{s} \TrMPack{(p, c)/(p_i, c_i)_{i=1..k}} \tilde{s}'$ exists if (one of) the sequence of operations applied on $\tilde{s}$ when packet $p$ arrives on channel $c$ outputs $(p_i, c_i)_{i=1..k}$ and leads to $\tilde{s}'$.

The semantics of operations is defined similarly to the ``standard'' relation effect semantics.
The semantics of error and output actions (that do not change the middlebox state) is straightforward.
Next, we explain the semantics of the operations that depend on or change the middlebox state --- membership queries and relation updates.

Consider a membership query $q$. Let $\tilde{s}$ be the middlebox state before
evaluating $q$, i.e., $\tilde{s}$ is the state that results from executing all
previous relation updates, and let $p$ be the packet that invoked the middlebox
transition. Then $q$ is evaluated to $\tilde{s}(p)(q)$.

\ignore{
Finally, we describe the semantics of the relation update operation.
Recall that middlebox programs allow only one kind of operation that may change the middlebox state, namely updating the membership state of some value in a relation.
We will now present the concrete semantics of such operations. 
}

%


\ignore{
Consider the operation \verb|r.insert(|$\overline{d}$\verb|)|, and let the packet the middlebox program is operating on be $p$.
}

Next, consider a relation update.
A relation update $\text{r}(\overline{a})~ \textbf{:=} \textit{cond}$ updates the packet states of all packets that are affected by the operation. 
This is done as follows.
As before, let $\tilde{s}$ be the intermediate state of $m$ right before executing the operation, and let $p$ be the packet that the middlebox program is operating on.
Consider the case where $\textit{cond}$ evaluates to $\True$ in $\tilde{s}$, corresponding to addition of a value. (Removal of a value is symmetric.)
We denote by $\overline{a}(p)$ the result of substituting each field name in $\overline{a}$ by its value in $p$. 
That is, $\overline{a}(p) \in D(\text{r})$ is the value being added to $\text{r}$. 
This addition may affect the value of membership queries $q \in Q(m)$ with $\rel(q) = \text{r}$ (querying the same relation $\text{r}$) for other packets $\tilde{p}$ as well, in case that $\atoms(q)(\tilde{p})$, i.e., the value being queried on $\tilde{p}$, is the same as the value $\overline{a}(p)$ being added to $\text{r}$.
Therefore, the intermediate state obtained after the relation update operation has been applied is
\begin{equation*}
  \tilde{s}'=\lambda \tilde{p} \in \packets.\ \lambda q \in Q(m).\ \\
  \begin{cases}
    \True, & \text{if $\rel(q)=$ r $\land$} 
    \atoms(q)(\tilde{p}) = \overline{a}(p).\\
    \tilde{s}(\tilde{p})(q), & \text{otherwise}.
  \end{cases}
\end{equation*}
Namely, the operation updates to $\True$ the value of queries that coincide with the tuple of elements inserted to the relation.

\ifexamples
\else
\begin{example}
\label{examp:packet-transition}
Consider the packet effect state
$\tilde{s}\eqdef\lambda p.\ \lambda q. \False \in \PState[\fw]$ of the firewall (\figref{session-fw-code}), where $q$ ranges over the three membership queries in the code.
Upon reading the packet $(h_1, h_2, 0)$ from an internal port,
the middlebox performs a sequence of internal transitions which includes evaluating the expression ``\texttt{p.type=0}'' to \True,
outputting the packet $(h_1,h_2,0)$ to the output port, and executing the command $\texttt{requested(p.dst) := true}$, which results in
updating the state to:

$\tilde{s}'\eqdef\lambda \tilde{p}.\ \lambda q.\ \left\{
                                                             \begin{array}{ll}
                                                               \True, & \text{if }\rel(q)=\requested \land \atoms(q)(\tilde{p})=h_2\hbox{} \\
                                                               \False, & \hbox{otherwise.}
                                                             \end{array}
                                                           \right.$

That is,
$\tilde{s}'((h_2,*,*))(\texttt{p.src in requested}) = \True$ and all the other values in $\tilde{s}'$ remain $\False$ as before.
Therefore, $\tilde{s} \TrMPack{((h_1,h_2,0),\overset\rightarrow{c_{in}})/((h_1,h_2,0),\overset\rightarrow{c_{out}})} \tilde{s}'$.
\qed
\end{example}
\fi


\ignore{
\para{Locality of middlebox transitions}
Note that the update of the packet state, $\tilde{s}(\tilde{p})$, of each packet $\tilde{p}$, performed by executing an operation $\text{r}(\overline{a})~ \textbf{:=} \textit{cond}$,
depends only on $\tilde{s}(\tilde{p})$, the input channel $c$, the input packet $p$  and $\tilde{s}(p)$ which determines the value of queries.
It is completely independent of the state of other packets. This means that the state updates are local.
Since, in addition, the execution path of the middlebox given input packet $p$ depends only on the packet state $\tilde{s}(p)$ of $p$, this locality extends to the entire middlebox program as well. Formally:

\begin{definition}
The transition relation $\TrMPack{(p, c)/o} \subseteq \PState[m] \times \PState[m]$ is \emph{local} if there exist
$\postLS: (Q(m) \to \{\True,\False\}) \times \packets \times E \times \packets \times (Q(m) \to \{\True,\False\}) \times {\N}^{*} \to (Q(m) \to \{\True,\False\})$
and $\postLP: (Q(m) \to \{\True,\False\}) \times \packets \times E \times {\N}^{*} \to (\packets\times E)^{*}$  such that
$s \TrMPack{(p, c)/o} s'$ iff there exists $\bar{n} \in {\N}^{*}$ such that
$s' =  \lambda \tilde{p}.\, postLS(s(p),p,e,\tilde{p},s(\tilde{p}),\bar{n})$ and $o = postLP(s(p),p,e,\bar{n})$.
\end{definition}

The $\postLS$ operator receives as input the packet state $s(p)$ of the input packet $p$ arriving on channel $c$, as well as the packet state $s(\tilde{p})$ of an arbitrary packet $\tilde{p}$, and returns the updated packet state $s'(\tilde{p})$ of $\tilde{p}$.
The $\postLP$ operator receives the packet state $s(p)$ of the input packet $p$ arriving on channel $c$, and returns the sequence of output packets.
To account for nondeterminism, both $\postLS$ and $\postLP$ take as a parameter a vector $\bar{n} \in {\N}^{*}$ that indicates which option is taken whenever there is a nondeterministic choice.
The locality of AMDL programs, ensures that both $\postLS$ and $\postLP$ can be computed in time linear in the size of the middlebox program.
This property will be important later to efficiently compute an abstract transformer (\secref{cartesian-packet}).
}

\ignore{
In order to define locality, we first express the transition relation of middlebox $m$ via a \emph{deterministic} $\post$ operator which receives the middlebox state, an input packet and an input channel, and computes the new state and a sequence of output packets. To account for nondeterminism, $\post$ takes as a parameter a vector $\bar{n} \in {\N}^{*}$ that indicates which option is taken whenever there is a nondeterministic choice.

\begin{definition}
The operator $\post: \PState[m] \times \packets \times E \times {\N}^{*} \to \PState[m] \times (\packets \times E)^{*}$ is \emph{local} if there exist
$\postLS: (Q(m) \to \{\True,\False\}) \times \packets \times E \times \packets \times (Q(m) \to \{\True,\False\}) \times {\N}^{*} \to (Q(m) \to \{\True,\False\})$
and $\postLP: (Q(m) \to \{\True,\False\}) \times \packets \times E \times {\N}^{*} \to (\packets\times E)^{*}$  such that
$\post(s, (p,c), \bar{n}) = \langle \lambda \tilde{p}.\, postLS(s(p),p,e,\tilde{p},s(\tilde{p}),\bar{n}), postLP(s(p),p,e,\bar{n}) \rangle$.
\end{definition}

The $\postLS$ operator receives as input the packet state $s(p)$ of the input packet $p$ arriving on channel $c$, as well as the packet state $s(\tilde{p})$ of an arbitrary packet $\tilde{p}$, and returns the updated packet state $s'(\tilde{p})$ of $\tilde{p}$.
The $\postLP$ operator receives the packet state $s(p)$ of the input packet $p$ arriving on channel $c$, and returns the sequence of output packets.
The locality of AMDL programs, ensures that both $\postLS$ and $\postLP$ can be computed in time linear in the size of the middlebox program.
This property will be important later to efficiently compute an abstract transformer (\secref{cartesian-packet}).
}

\subsection{Bisimulation of Packet Effect Semantics and Relation Effect Semantics}
\label{sec:bisim}
We continue by showing that the transition systems defining the semantics of middleboxes in the packet effect and in the relation effect representations are bisimilar.

To do so, we first define a mapping $\ps \colon \RState[m] \to \PState[m]$ from the relation state representation to the packet effect state representation.
Recall that the relation state representation of middlebox states is $s \in \RState[m]\eqdef \rels(\mboxe) \to \pow{D(m)}$.
Given a state $s \in \RState[m]$, $\ps$ maps it to the packet effect state $s^\Pack$ defined as follows:
\[
s^\Pack \eqdef \lambda \tilde{p} \in \packets.\ \lambda q \in Q(m). \
\atoms(q)(\tilde{p}) \in s(\rel(q)).
\]
That is, for every input packet $\tilde{p}$, the value in $s^\Pack$ of the query $q \in Q(m)$
is equal to the evaluation of the same query in $s$ based on an input packet $\tilde{p}$.



\ignore{
\para{Relation Space to Packet Space}
First, we show how to obtain a state represented in the packet space from a state represented in the relation space to.

Recall that $S: R\to D_R \to \{T,F\}$ is the relation space representation of a middlebox state. Let $s\in S$ be the state of a middlebox. Then $s'=\lambda p\in \packets.\lambda q \in Q. s(q.relation)(q.value[p.fields/fields])$ is the equivalent packet space state. We denote this state translation by $s \rightarrow_p s'$.
}

\begin{definition}[Bisimulation Relation]
For a middlebox $m$, we define the relation $\sim_m \subseteq \RState[m] \times \PState[m]$ as the set of all
pairs $(s, s^p)$ such that  
$s = s^p = \err$ or $\ps(s)  = s^p$.
\end{definition}



\begin{lemma}
\label{lem:bisimulation}
Let $s\in \RState[m]$ and $\tilde{s} \in \PState[m]$ and $s \sim_m \tilde{s}$. Then the following holds:
\begin{itemize}
\item For every state $s' \in \RState[m]$, if  $s \TrMRel{(p, c)/o} s'$ then there exists a state $\tilde{s}' \in \PState[m]$ s.t. $\tilde{s}  \TrMPack{(p, c)/o} \tilde{s}'$ and $s' \sim_m \tilde{s}'$, and
\item For every state $\tilde{s}' \in \PState[m]$ if $\hat{s} \TrMPack{(p, c)/o} \tilde{s}'$ then there exists a state $s'\in \RState[m]$ s.t. $s \TrMRel{(p, c)/o} s'$ and $s' \sim_m \tilde{s}'$.
\end{itemize}
%
\end{lemma}

\ignore{
\subsection{Back to Network Semantics}

\TODO{SH: what do we do with this subsection? We just want the notation. Everything else is already defined. Ideally, if the network semantics is parametric we can just minimize this entire subsection to a lemma (or corollary) that lifts the bisimulation to the network level (with every possible semantics.}

By plugging-in the two representations of middleboxes in the definition of the network semantics, we obtain two variants of the network semantics.

\TODO{SH: revise}
\para{Packet State Network Configurations}
The packet state network configurations are defined as $(\sigma, \pi) \in \PState = (M\to \PState[M])\times(E \to \packets^{*})$, where
$\PState[M] = \bigcup\limits_{m\in M}\PState[m]$ denotes the
set of middlebox states in the packet state representation of all middleboxes in the network.
Error configurations, denoted $\err$, are defined as in \defref{error-conf}.

\para{Initial Configuration}
The initial packet state configuration is $(\PState(\sigma_I),\lambda\,e\in E\,.\, \epsilon)$.

\para{Packet State Network Transitions}
The four semantics of network transitions presented in \secref{network-semantics} are naturally adapted to packet state configurations, the only difference is that the middlebox state transitions are over the packet space, as described in \secref{packet-space-concrete-transformer}. The collecting semantics are also adapted. Of special interest are the collecting packet state ordered semantics, denoted ${\BBsub{\Net}{pa}^{o}}$, and the collecting packet state unordered reverting semantics, denoted ${\BBsub{\Net}{pa}^{ur}}$. 



The bisimulation between middlebox representations can be lifted to a bisimulation between each relation state network semantics and the corresponding packet state network semantics.
Therefore, the following holds:

\begin{lemma}
For every semantic identifier $i \in \{o,u,or,ur\}$, $\err \in {\BB{\Net}^{i}}$ if and only if $\err \in {\BBsub{\Net}{pa}^{i}}$.
\end{lemma}

\TODO{SH: if we lift the bisimulation relation to network states we can state a stronger lemma that says that the configurations in the collecting semantics are bisimilar. Then the above lemma can be a corollary.}
}

\subsection{Locality of Packet-Effect Middlebox Transitions}
\label{sec:local-sem}

In this section we present a locality property of the packet effect semantics that will allow us to efficiently compute an abstract transformer when applying a Cartesian abstraction. 
%
Namely, we observe that an execution of an operation $\text{r}(\overline{a})~ \textbf{:=} \textit{cond}$,
in the context of processing an input packet $p$,
potentially updates the packet states of all packets.
However, for each packet $\tilde{p}$, the updated packet state $\tilde{s}'(\tilde{p})$
depends only on its pre-state $\tilde{s}(\tilde{p})$, the input channel $c$, the input packet $p$,
and $\tilde{s}(p)$, which determines the value of queries;
it is completely independent of the packet states of all other packets.
Since, in addition, the execution path of the middlebox when processing input packet $p$ depends only on the packet state of $p$, this form of \emph{locality}, which we formalize next, extends to entire middlebox programs. 

\ignore{
Recall that each middlebox $m$ is associated with a program, $\prog(m)$,
given by an abstract syntax tree (AST). An evaluation of a program $\prog(m)$, for a given input packet,
is performed by a single recursive traversal of the AST (as there are no looping constructs) where some of the
branches non-deterministically choose a child node.
We can make the evaluation deterministic by adding a vector $\bar{n} \in {\N}^{*}$ as a parameter,
to indicate the child nodes to be taken at non-deterministic branches.
The vector is taken from the finite set $\paths(m)\eqdef\{n_1\ldots\ n_k \mid k=|\prog(m)| \land 1\leq n_i \leq k\}$.
We write 
$\tilde{s} \TrMPackM{(p, c)/(p_i, c_i)_{i=1..k},\bar{n}} \tilde{s}'$
if $\tilde{s} \TrMPackM{(p, c)/(p_i, c_i)_{i=1..k}} \tilde{s}'$ is a transition corresponding to the evaluation
of $\prog(m)$ with the deterministic choices given by the vector $\bar{n}$.
Note that in \AMDL, the path in the AST taken when evaluating $\prog(m)$ on input packet $p$ depends only on $p$, the packet-state of $p$, and $\bar{n}$.

\cbstart
\begin{example}
In the previous example, the transition from $\tilde{s}$ to $\tilde{s}'$ follows the path $\bar{n}=1, 2$.\qed
\end{example}
}


\begin{definition}[Substate]
Let $\tilde{s} \in P\rightarrow Q(m)\rightarrow\{\True,\False\}$ be a packet effect state. We denote by $\tilde{s}|_{\{p,\tilde{p}\}} \in \{p,\tilde{p}\}\rightarrow Q(m)\rightarrow\{\True,\False\}$ the \emph{substate}
obtained from $\tilde{s}$ by dropping all packet states other than those of $p$ and $\tilde{p}$.
Let $\PState[m,p,\tilde{p}] \eqdef \{p,\tilde{p}\}\rightarrow Q(m)\rightarrow\{\True,\False\}$ denote the set of
substates for 
$p$ and $\tilde{p}$.
\end{definition}

\begin{definition}[Substate transition relation]
We define the \emph{substate transition relation}\\
$\TrMPackSub{(p, c)/(p_i, c_i)_{i=1..k}} : \PState[m,p,\tilde{p}] \times \PState[m,p,\tilde{p}]$ as follows.
A \emph{substate transition}\\
 $\tilde{s}[p,\tilde{p}] \TrMPackSub{(p, c)/(p_i, c_i)_{i=1..k}} \tilde{s}[p,\tilde{p}]'$ holds if there exist $\tilde{s}$ and $\tilde{s}'$ such that
$\tilde{s}|_{[p,\tilde{p}]} =  \tilde{s}[p,\tilde{p}]$, $\tilde{s}'|_{[p,\tilde{p}]} = \tilde{s}[p,\tilde{p}]'$ and
$\tilde{s} \TrMPack{(p, c)/(p_i, c_i)_{i=1..k}} \tilde{s}'$.
\end{definition}

The locality of \AMDL\ programs manifests itself in the ability to compute the substate transition relation, $\TrMPackSub{(p, c)/(p_i, c_i)_{i=1..k}}$, directly from the code (without first computing the transition relation and then using projection). This property will be important later to efficiently compute a network-level abstract transformer (\secref{cartesian-packet}):

\begin{lemma}[2-Locality]
Given $\tilde{s}[p,\tilde{p}]$ and $\tilde{s}[p,\tilde{p}]'$, checking whether 
\begin{align*}
\tilde{s}[p,\tilde{p}] \TrMPackSub{(p, c)/(p_i, c_i)_{i=1..k}} \tilde{s}[p,\tilde{p}]'
\end{align*}
can be done in time linear in the size of the middlebox program.
\end{lemma}

\ignore{
\begin{example}
Continuing with the definitions in \exref{substate} and $\bar{n}=1,2$,
we have the following substate transition for $\fw_1$:
\[
\tilde{s}|_{\{p,\tilde{p}\}} \TrMPackSub{(p, \overset\rightarrow{e_2})/(p, \overset\rightarrow{e_3}), \bar{n}} \tilde{s}'|_{\{p,\tilde{p}\}} \enspace.
\]
Notice that in this case, since no query was made, the output substate is obtained by
propagating the packet state of $p$ and updating the packet state of $\tilde{p}$, independent of any other packet states.
\qed
\end{example}

\begin{definition}[2-Local Middlebox Semantics]
A middlebox semantics is \emph{2-local} if, for every packet $p$,
the following holds:
\[
\tilde{s} \TrMPack{(p, c)/(p_i, c_i)_{i=1..k}, \bar{n}} \tilde{s}' \;\;\Longleftrightarrow\;\;
\forall  \tilde{p}  \;.\; \tilde{s}|_{\{p,\tilde{p}\}} \TrMPackSub{(p, c)/(p_i, c_i)_{i=1..k}, \bar{n}} \tilde{s}'|_{\{p,\tilde{p}\}}
 \enspace.
\]
\end{definition}
\cbend

\begin{theorem}
The \AMDL\ packet effect semantics \iflong(\appref{amdl-semantics}) \fi is 2-local.
\end{theorem}

Furthermore, the substate transition relation, $\TrMPackSub{(p, c)/(p_i, c_i)_{i=1..k}, \bar{n}}$, can be computed directly from the code (without first computing the transition relation and then using projection)
in time linear in the size of the middlebox program. This is an even stronger notion of locality, which will be important later to efficiently compute a network-level abstract transformer (\secref{cartesian-packet}).
} 
\section{Network Semantics}
\label{sec:network-semantics}

This section defines the 
semantics of stateful networks by defining the semantics of packet traversal over communication channels in the network, and the transitions between network configurations.
We first define a concrete semantics, followed by two relaxations: unordered semantics and reverting semantics.
These relaxations provide sufficient conditions for completeness of the abstract interpretation performed in 
\secref{abs-int}.

\figref{CombinedAbstractions} provides a high-level view of the different network semantics.

\para{Network Topology}
A \emph{network} $\Net$ is a finite bidirected\footnote{A \emph{bidirected
graph} is a directed graph in which every edge has a matching edge in the
opposite direction. i.e., $(u,v)\in E \iff (v,u)\in E$.} graph of \emph{hosts}
and \emph{middleboxes}, equipped with a \emph{packet domain}. Formally, $\Net =
(\hosts\cup \mboxes, E, \packets)$, where:

\begin{itemize}
\item $\packets$ is a set of packets.
\item $\hosts$ is a finite set of \emph{hosts}. A \emph{host} $h \in H$
consists of a unique identifier and a set of packets $\packets_h \subseteq
\packets$ that it can send.
\item $\mboxes$ is a finite set of \emph{middleboxes}. A middlebox $m\in
\mboxes$ is associated with a set of communication channels $C_m$.
\item $E \subseteq \{\gV{h,c_m,m},\gV{m,c_m,h}\mid h\in H ,m\in M, c_m\in C_m\}
\cup \{\gV{m_1,c_{m_1},c_{m_2},m_2}\mid m_1,m_2\in M, c_{m_1}\in
C_{m_1},c_{m_2}\in C_{m_2}\}$ is the set of directed communication channels in
the network, each connecting a communication channel $c_{m_1}\in C_{m_1}$ of
middlebox $m_1$ either to a host, or to a communication channel $c_{m_2}\in
C_{m_2}$ of middlebox $m_2$. For $e$ of the form $\gV{m,c_m,h}$ or
$\gV{m,c_{m},c_{m_2},m_2}$, we say that $e$ is an \emph{egress} channel of
middlebox $m$ connected to channel $c_m$ and an \emph{ingress} channel of host
$h$, respectively middlebox $m_2$, connected to channel $c_{m_2}$.

\end{itemize}

%



\ignore{
\para{Reverting Middleboxes}
We propose two network semantics: the \emph{FIFO network semantics} (\secref{fifo-semantics}) preserves the packet transmission order, \emph{whereas unordered network semantics} (\secref{unordered-semantics}) does not.
In both semantics middleboxes may non-deterministically revert to their initial state following a packet processing event.
}


The network semantics is parametric in the middlebox semantics.  
It considers the semantics of a middlebox $\mboxe\in \mboxes$ to be a transition system with a finite set of states $\Sigma[\mboxe]$, an initial state $\sigma_I(m) \in \Sigma[\mboxe]$ and a set of transitions $\TrM{(p, c)/(p_i, c_i)_{i=1..k}} \subseteq \Sigma[m] \times \Sigma[m]$. 
This can be realized with either the relation effect semantics or the packet effect semantics defined in \secref{relation-sem} and \secref{packet-sem}, respectively.

\ignore{
\para{Safety} Safety properties are expressed by middleboxes that might transition to an error state.
}
\ignore{
Formally, we introduce a unique \emph{error} middlebox state, denoted $\err$.
We assume that $\err \in \Sigma[\mboxe]$ for every middlebox $\mboxe$.
}

%

\subsection{Concrete (Ordered) Network Configurations}

All variants of the network semantics defined in this section are defined over the same set of configurations.
Let $\Sigma[M] \eqdef \bigcup\limits_{m\in M}\Sigma[m]$ denote the
set of middlebox states of all middleboxes in a network.
An \emph{ordered network configuration} $(\sigma, \pi) \in \Sigma = (M\to \Sigma[M])\times(E \to \packets^{*})$
assigns middleboxes to their (local) middlebox states and
communication channels to sequences of packets.
The sequence of packets on each channel represents all packets sent from the source and not yet processed by the destination.


\para{Initial Configuration}
We denote the ordered initial configuration by $(\sigma_I,\lambda\,e\in E\,.\, \epsilon)$,
where $\sigma_I \colon M\to \Sigma[M]$ denotes the initial state of all middleboxes.


\para{Error Configurations}
We say that a configuration is an \emph{error configuration} if any of its
middleboxes is in the error state.
We denote all error configurations by $\err$.


\subsection{Concrete (FIFO) Network Semantics}
\label{sec:fifo-semantics}

We first consider the First-In-First-Out (FIFO) network semantics, under which communication channels retain the order in which packets were sent.
%

\para{Ordered Network Transitions}
The network semantics is defined via \emph{middlebox transitions} and \emph{host transitions}.

A middlebox transition is $(\sigma,\pi) \TrN{p,e,m}{o} (\sigma',\pi')$ where
the following holds:
(i) $p$ is the \emph{first} packet on the channel $e \in E$,
(ii) the channel $e$ is an ingress channel of middlebox $m$ connected to channel $c \in C_m$, 
(iii) $\sigma(m) \TrM{(p, c)/(p_i, c_i)_{i=1..k}} \sigma'(m)$, meaning that $\sigma'(m)$ is the result of updating $\sigma(m)$ according to the middlebox semantics,
(iv) the channels $e_i$ are egress channels of middlebox $m$ connected to the channels $c_i \in C_m$, 
(v) $\pi'$ is the result of removing packet $p$ from (the head of) channel $e$ and appending $p_i$ to the tails of the appropriate channels $e_i$, and
(vi) the states of all other middleboxes equal their states in $\sigma$.


A host transition is $(\sigma,\pi) \TrN{h,e,p}{o} (\sigma,\pi')$ where one of the following holds:
\begin{description}
\item[Packet Production]
(i) the channel $e$ is an egress channel of host $h$,
(ii) $p\in P_h$ is a packet sent by $h$, and
(iii) $\pi'$ is the result of appending $p$ to the tail of $e$; or

\item[Packet Consumption]
(i) the channel $e$ is an ingress channel of host $h$,
(ii) $p$ is the first packet on the channel $e$, and
(iii) $\pi'$ is the result of removing $p$ from the head of $e$.
\end{description}

We denote the ordered transition relation obtained by the union of all middlebox and host transitions by $\TrN{}{o}$.
\iflonglong
\[
\TrN{}{o} =  \left( \bigcup_{p \in \packets,e \in E,m \in \mboxes} \TrN{p,e,m}{o} \right) \cup \left( \bigcup_{h \in H, p \in \packets, e \in E} \TrN{h,e,p}{o} \right).
\]
\fi
It is naturally lifted to a concrete transformer $\Tr^o \colon \pow{\Sigma} \to \pow{\Sigma}$ defined as:
\[
\Tr^o(X) \eqdef \{ (\sigma', \pi') \mid (\sigma, \pi) \in X \land (\sigma, \pi) \TrN{}{o} (\sigma', \pi') \} \enspace.
\]

\para{Collecting Semantics} The ordered collecting semantics of a network $\Net$ is the set of configurations reachable from the initial configuration.
\[
\begin{array}{rclrcl}
\BB{\Net}^o &\eqdef& \lfp(\Tr^o)(\sigma_I,\lambda\,e\in E\,.\, \epsilon) 
            &=& \bigcup\limits_{i=1}^\infty (\Tr^o)^i (\sigma_I,\lambda\,e\in E\,.\, \epsilon) \enspace.
\end{array}
\]


\begin{definition}[Safety Verification Problem]
%
For a network $\Net$ and initial state $\sigma_I$ for the middleboxes, the safety verification problem is to determine whether an error configuration is reachable from the initial configuration. 
That is, whether $\err \in \BB{\Net}^o$.
\end{definition}

\begin{theorem}~\cite{velner2016some}
\label{thm:OrderedUndecidable}
The safety verification problem for ordered networks
is undecidable.
\end{theorem}

In this work, we tackle the undecidability of verification by developing a sound abstract interpretation that can be used to check the safety of networks.
Before doing so, we present two relaxed network semantics that motivate the abstractions we employ, and also provide sufficient conditions for their completeness.

\subsection{Unordered and Reverting Network Semantics}
\label{sec:unordered-semantics}

The ``unordered'' semantics allows channels to not preserve the packet transmission order. Namely, packets in the same channel may be processed in a different order than the order in which they were received. The ``reverting'' semantics allows middleboxes to revert to their initial state after every transition. Formally, these relaxed semantics extend the set of network transitions (and consequently, the transformer and the collecting semantics) with reordering transitions and reverting transitions, respectively.

A \emph{reordering transition} has the form $(\sigma,\pi) \TrN{e}{} (\sigma,\pi')$ where for the channel $e \in E$, $\pi'(e)$ is a permutation of $\pi(e)$ and for all other channels $e' \neq e$, $\pi'(e') = \pi(e')$.

A \emph{reverting transition} has the form $(\sigma,\pi) \TrN{m}{} (\sigma',\pi)$ where for the middlebox $m \in M$, $\sigma'(m)=\sigma_I(m)$ and for all other middleboxes $m' \neq m$, $\sigma'(m) = \sigma(m)$.

The \emph{unordered network transitions} consist of the ordered transitions as well as the reordering transitions; 
the \emph{ordered reverting transitions}  consist of the ordered transitions and the reverting transitions; and 
the \emph{unordered reverting transitions} consist of all of the above. We denote the corresponding collecting semantics by $\BB{\Net}^u$, $\BB{\Net}^{or}$ and $\BB{\Net}^{ur}$, respectively. Clearly,
\[
\begin{array}{ccc}
{\BB{\Net}^{o}} \subseteq {\BB{\Net}^{u}} \subseteq {\BB{\Net}^{ur}}
&
\text{ and }
&
{\BB{\Net}^{o}} \subseteq {\BB{\Net}^{or}} \subseteq {\BB{\Net}^{ur}}
\end{array}
\]

By plugging-in the two representations of middleboxes in the definition of the network semantics, we obtain two variants of the network semantics for each of the four variants considered so far.
In the sequel, we use a $pa$ subscript to refer to the packet effect semantics, and no subscript to refer to the relation effect semantics.
%
The bisimulation between middlebox representations is lifted to a bisimulation between each relation state network semantics and the corresponding packet state network semantics.
Therefore, the following holds:

\begin{lemma} \label{lem:packet-effect-equiv}
For every semantic identifier $i \in \{o,u,or,ur\}$, $\err \in {\BB{\Net}^{i}}$ iff $\err \in {\BBsub{\Net}{pa}^{i}}$.
\end{lemma}

The safety verification problem is adapted for the different variants of the network semantics. 
The following theorem summarizes the complexity of the obtained problems. (We do not distinguish the packet effect semantics from the relation effect semantics, since due to \Cref{lem:packet-effect-equiv} they induce the same safety verification problem.)

\begin{theorem} \label{thm:reverting-unordered-conphard}
The safety verification problem is
\begin{compactenum}[(i)]
\item \label{item:comp-unordered} EXPSACE-complete for unordered networks~\cite{velner2016some}.
\item \label{item:comp-reverting} undecidable for ordered reverting networks (\appref{proofs}).
\item \label{item:comp-unordered-reverting} coNP-hard for unordered reverting networks (\appref{proofs}).
\end{compactenum}
\end{theorem}

\Cref{thm:reverting-unordered-conphard}(\ref{item:comp-reverting}) justifies the
need for the unordered abstraction even in reverting networks. \Cref{thm:reverting-unordered-conphard}(\ref{item:comp-unordered-reverting}) 
implies that our abstract interpretation algorithm, presented in \secref{abs-int}, which is both sound and complete for the unordered reverting semantics, is essentially optimal since it essentially meets the lower bound stated in the theorem (it is exponential in the number of state queries of any middlebox and polynomial in the number of middleboxes, hosts and packets).

%



%

\ignore{
The reverting unordered semantics is
not directly implementable (the set of packet multisets is infinite).
We continue by proving and utilizing properties of the semantics,
which in later sections allow us to apply sound abstractions, while still maintaining
fixpoint-completeness, leading us to the
\emph{Cartesian network abstraction}.
}


\para{Sticky Properties}
\cbstart
Unordered reverting networks have a useful property of \emph{sticky packets}, meaning that if a packet is pending for a middlebox in some run of the network then any run has an extension in which the packet is pending again with multiplicity $> n$, for any $n \in \Nat$.
This property implies a stronger property: 

%

\ignore{
two useful properties, which we define next --- the \emph{sticky states} property and the \emph{sticky packets} property --- that make the abstract interpretation described in \secref{abs-int} complete. \cbend
\ignore{
Intuitively, the sticky packets property states that if a packet is pending for a middlebox in some run of the network then any run has an extension in which the packet is pending again.
Consequently, it is not necessary to maintain the cardinality of every packet on every channel and
it is not necessary to maintain the correlation between a channel and the set of packets
that can simultaneously appear there.
Instead, it is sufficient to record, for each individual channel $e$ and individual packet $p$,
whether $p$ has arrived on $e$ at some point.
The sticky states property ensures that ignoring the correlation between the state of a middlebox $m$ and the channels connected to it does not lose precision (for safety).
This implies that one can also ignore correlations between the states of different middleboxes, making a Cartesian abstraction precise.
}


\begin{lemma}[Sticky Packets Property]
\label{lem:sticky-packets}
For every channel $e$ and packet $p$:
If in some reachable configuration $e$ contains $p$, then every run can be extended such that $e$ will eventually contain $p$.
Moreover, every run can be extended such that $e$ will eventually contain $n$ copies of $p$ (for every $n>0$).
\end{lemma}
%
%

The sticky packets property implies that, in verifying such networks, it is not necessary to maintain the cardinality of every packet on every channel and
it is not necessary to maintain the correlation between a channel and the set of packets
that can simultaneously appear there.
Instead, it is sufficient to record, for each individual channel $e$ and individual packet $p$,
whether $p$ has arrived on $p$ at some point, without losing precision w.r.t. safety.
}


\ignore{
\begin{lemma}[Sticky States Property]
\label{lem:sticky-states}
For every channel $e$, packet $p$, middlebox $m$ and state 
$s$ of $m$:
If, in some reachable configuration, channel $e$ contains $p$ and in some (possibly other) reachable configuration $m$ is in state 
$s$, then 
there exists a reachable configuration where simultaneously $e$ contains $p$ and $m$ is in state $s$. 
\end{lemma}

%
%

The sticky states property ensures that ignoring the correlation between the state of a middlebox $m$ and the channels connected to it does not incur any precision loss w.r.t. safety.
This implies that one can also ignore correlations between the states of different middleboxes, making the network-level Cartesian abstraction defined in \secref{abs-int} precise.
}


\ignore{
The sticky properties formulated above are oblivious of the middlebox semantics. For the case of the packet state representation of the middlebox semantics, an even finer variant of the sticky states property of reverting unordered networks arises, which implies precision of the middlebox-level Cartesian abstraction defined in \secref{abs-int}:
}
\begin{lemma} [Sticky Packet States Property]
\label{lem:sticky-packet-states}
For every channel $e$, packets $p, \tilde{p}$, middlebox $m$ and packet state $\tilde{v}$ of $\tilde{p}$ in $m$:
If, in some reachable configuration, channel $e$ contains $p$ and in some (possibly other) reachable configuration the packet state of $\tilde{p}$ in $m$ is $\tilde{v}$, then 
there exists a reachable configuration where simultaneously $e$ contains $p$ and the packet state of $\tilde{p}$ in $m$ is $\tilde{v}$.
\end{lemma}
\ignore{
\iflong
\begin{proof}
Follows directly from \lemref{sticky-states} and from the bisimulation between the relation state space and the packet state space of middleboxes.
\end{proof}
\fi
}

Intuitively,\ignore{ \lemref{sticky-packets}, \lemref{sticky-states} and} \lemref{sticky-packet-states} follows from the fact that all middleboxes can revert to their initial state and the unordered semantics enables a scenario where the particular state and packets are reconstructed.
It ensures that ignoring the correlation between the packet states of a middlebox for different packets, the packet states across different middleboxes, and the occurrence (and cardinality) of packets on channels does not incur any precision loss w.r.t. safety. This makes the network-level abstraction defined in \secref{abs-int}, which treats channels as sets of packets and ignores correlations between packet states and channels, precise.

\ignore{
The sticky packet states property ensures that ignoring the correlation between the packet states of a middlebox $m$ and the channels connected to it does not incur any precision loss w.r.t. safety. This implies that one can also ignore correlations between the different packet states of a middlebox as well as of other middleboxes, making the network-level Cartesian abstraction defined in \secref{abs-int} precise.
}

\ignore{

\subsubsection{Undecidability of Safety of Reverting Networks Over Ordered Channels}
\label{sec:ordered-undecidability}
\TODO{where should this be?}

We show that if the network channels are FIFO, then we can simulate the behavior of a non-reverting network with a reverting network.
As the isolation problem for the former problem is known to be undecidable~\cite{velner2016some}, we conclude undecidability for the latter.

Intuitively our reduction makes sure that whenever one middlebox reverts the entire network reverts, i.e., all the middleboxes and all the pending packets in the channel also revert.

For this purpose we add an auxiliary \emph{wake-up} packet types and auxiliary \emph{wake-up} states for every middlebox.
Every middlebox initial state is a wake-up state, and in this state the middlebox only forwards wake-up packets and ignores (discard without any state change) other packet types.
When it forwards the wake-up packet it changes the packet type so it will also denotes the fact that the specific middlebox is in a wake-up state.
The middlebox exits the wake-up states only after it gets confirmations that all the other middleboxes in the network are in a wake-up state.

We observe that since the network is FIFO it must be the case that there are no non-wake-up packets in the channels between the middleboxes.
After it leaves the wake-up states it goes to the original initial state of the middlebox in the non-reverting network.

In this reduction, the reverting property does not change the reachability of error states, hence we get the following theorem:

\begin{theorem}
The safety verification problem over reverting FIFO networks is undecidable.
\end{theorem}

\subsubsection{Isolation for Reverting Networks with Arbitrary Number of Queries is coNP-hard}
\TODO{move somewhere. Not sure where. But needs to be after we define queries|| Why? we define queries in the syntax.}
\label{sec:conphardness}
We prove that if the number of queries in a middlebox is not fixed, then the isolation problem is coNP-hard even for one middlebox and one host.
The proof is by reduction from UNSAT.
Given a formula $\phi$ with $n$ variables $x_1,\dots,x_n$ we construct a network with one host and one middlebox $m$, such that $m$ has only one port.
The packet tag set is $x_1,\neg x_1, \dots, x_n, \neg x_n$, i.e., there are $2n$ packet types.
The middlebox has two nullary relations $O_i,V_i$ for every $i\in\{1,\dots,n\}$, where intuitively, $O_i$ indicates whether a packet of type $x_i$ or $\neg x_i$ already occurred and $V_i$ indicates if the value of the first such packet is true or false.
Initially all the relations are initialized to false.
The relation $V_i$ is updated only if $O_i$ is false, and then $O_i$ is also updated.
The middlebox discards all packets.
In addition, whenever the interpretation of $O_i$ and $V_i$ satisfies $\phi$ an abort state is reached.
Clearly, the size of the code of $m$ is polynomial and isolation is violated if and only if $\phi$ is satisfiable.
We note that possible resets do not affect the isolation property.
\begin{theorem}
The isolation problem for reverting unordered networks is coNP-hard.\footnote{The safety verification
problem and the isolation problem are equivalent.}
\end{theorem}

}

\section{Abstract Interpretation for Stateful Networks}
\label{sec:abs-int}

In this section, we present our algorithm for safety verification of stateful networks based on abstract interpretation of the semantics ${\BBsub{\Net}{pa}^{o}}$, and discuss its guarantees.
\iflonglong
As explained in \secref{Intro},
applying an order abstraction combined with Cartesian abstraction on the concrete network domain allows us to obtain a sound abstract interpretation which is polynomial in the number of middlebox states.
Unfortunately, the latter is often exponential in the size of the packet space, which hinders scalability.
However, when also employing Cartesian abstraction on the packet effect representation of middlebox states, we obtain an abstract domain
for which the time needed to compute the abstract least fixpoint is polynomial
in the number of packets and the number of middleboxes, but exponential in the maximal number of distinct state queries in the middlebox programs.
Further, the abstract interpretation is sound for any network, and complete for reverting unordered networks.
\fi

\ignore{
To improve the complexity, we present an equivalent network semantics on an alternative representation of middlebox states. This semantics has the advantage that it facilitates an additional abstraction of the middlebox states without incurring any additional precision loss w.r.t. safety.
The result of employing the order abstraction and Cartesian abstraction on the new concrete domain is an abstract domain
in which the abstract interpretation is sound for any network, and complete for reverting unordered networks.
Further, the time needed to compute the abstract least fixpoint is polynomial
in the number of packets and the number of middleboxes, but exponential in the maximal number of conditions in the middlebox programs.
}

\subsection{Abstract Interpretation for Packet Space} \label{sec:cartesian-packet}
\ignore{
We now present our safety verification algorithm as an abstract interpretation
of the semantics ${\BBsub{\Net}{pa}^{o}}$. (In fact, our algorithm is sound and complete for ${\BBsub{\Net}{pa}^{ur}} \supseteq {\BBsub{\Net}{pa}^{o}}$.)
\DONE{SH: shouldn't it be ${\BBsub{\Net}{pa}^{o}}$? i.e., isn't it an abstract interpretation of the concrete semantics? (which happens to be ``equivalent'' to the unordered reverting semantics).}{I changed it.}

%
}

We apply sound abstractions to different components of the concrete packet state network domain.
\ignore{
to obtain different abstract interpretations. The abstractions apply either to the channel contents,
to the middlebox states, or to the correlations between these individual components.
}
Due to space constraints, we do not describe the intermediate steps in the
construction of the abstract domain, and only present the final domain used by
the analysis. Roughly speaking, the obtained domain abstracts away
\begin{inparaenum}[(i)]
\item the order and cardinality of packets on channels;
\item the correlation between the states of different middleboxes and different
channel contents; and
\item the correlation between states of different packets within each middlebox.
\end{inparaenum}


\newcommand{\absdom}{\mathcal{A}}
\newcommand{\concdom}{\mathcal{C}}
\newcommand{\absorder}{\sqsubseteq}

\para{Cartesian Packet Effect Abstract Domain}
Let $Q \to \{T,F\}$ denote the union of $Q(m) \to \{T,F\}$ over all middleboxes $m \in M$, including the error state $\err$.
The Cartesian abstract domain of the packet state of the network is given by the lattice $\absdom \eqdef (A, \bot, \absorder, \join)$, where
$A \eqdef (M \rightarrow \packets \to \pow{Q \to \{T,F\}}) \times (E \to \pow{\packets})$.
That is, an abstract element maps each packet in each middlebox to a set of possible valuations for the queries, and each channel to a set of packets.
The bottom element is $\bot \eqdef (\lambda m.\ \lambda p.\ \emptyset, \lambda e.\ \emptyset)$,
the partial order $a_1 \absorder a_2$ is defined by pointwise set inclusions per middlebox and channel, and join is defined by pointwise unions
$(\omega_1,\omega_2) \join (\omega_1',\omega_2')
\eqdef (\lambda m.\ \lambda p.\ \omega_1(m)(p) \cup \omega_1'(m)(p), \lambda e.\ \omega_2(p) \cup \omega_2'(p))$.

Let $\concdom \eqdef (\pow{\PState}, \subseteq)$ be the concrete network domain.
We define the Galois connection $(\concdom, \gamma, \alpha, \absdom)$ as follows.
The abstraction function $\alpha: \pow{\PState} \to A$ for a set of packet state configurations $X \subseteq \PState$
is defined as $\alpha(X) = (\omega_{\textit{mboxes}},\omega_{\textit{chans}})$ where
\[
\begin{array}{ccc}
\omega_{\textit{mboxes}} = \lambda m. \ \lambda p. \ \{ \sigma(m)(p) \mid (\sigma,\pi) \in X \}
& \text{and} &
\qquad \omega_{\textit{chans}} = \lambda e. \bigcup\limits_{(\sigma,\pi) \in X} \pi(e) \enspace.
\end{array}
\]

The concretization function $\gamma: A \to \pow{\PState}$ is induced by $\alpha$ and $\sqsubseteq$.
We denote the initial abstract element as $a_I = \alpha(\{(\sigma_I,\lambda\,e\in E\,.\, \emptyset)\})$.


%


\newcommand{\inchannel}{\textit{in}}
\newcommand{\outchannel}{\textit{out}}

\para{Abstract Transformer}
Next, we define the abstract transformer $\Tr^{\sharp}: A \to A$, which soundly abstracts the concrete transformer $\Tr^o$ 
and show that it is efficient, due to the locality property of middlebox transitions.
We use the predicate $\inchannel(c,e,m)$ to denote that the network channel $e$ is an ingress channel of middlebox $m$, connected to its
$c$ channel.
Similarly, $\outchannel(c,e,m)$ means that $e$ is an egress channel of $m$ connected to its $c$ channel.
%
%
%
\ignore{
\begin{definition}
Let $(\omega_1,\omega_2) \in (M \rightarrow     \packets \to \pow{Q \to \{T,F\}}) \times (E \to \pow{\packets})$ be an abstract element.
Then $\Tr^{\sharp}(\omega_1,\omega_2) = (\omega_1',\omega_2')$, where
\[
\begin{array}{l}
\omega_1'(m)(\tilde{p}) = 
\{\postLS_m(v,p,c,\tilde{p},\tilde{v},\bar{n}) \mid \exists c, e, p, v, \bar{n}. \ \inchannel(c,e,m) 
\wedge p \in \omega_2(e)
\wedge v \in \omega_1(m)(p) \wedge \tilde{v} \in \omega_1(m)(\tilde{p})  \}\\
\omega_2'(\tilde{e}) = 
\{p_i \mid \exists c, \tilde{c}, e, p, v, \bar{n}, m. \ (p_i,\tilde{c}) \in \postLP_m(v,p,c,\bar{n}) \wedge \inchannel(c,e,m) 
\wedge p \in \omega_2(e) \wedge v \in \omega_1(m)(p) \wedge \outchannel(c,\tilde{e},m) \} \enspace.
\end{array}
\]
\end{definition}
Intuitively, the transformer considers a packet state $v$ of some packet $p$ pending for $m$, i.e., $p \in \omega_2(e)$ and $v \in \omega_1(m,p)$.
It then computes the effect of this packet on (1) the packet state in $m$ of other packets $\tilde{p}$ with packet state $\tilde{v}$ in $m$ and on (2) the output channels. All the new packet states generated for $m$ 
are accumulated in $\omega_1'(m)(\tilde{p})$. Similarly, all the new packets generated
on channels $\tilde{e}$ are accumulated in $\omega_2'(\tilde{e})$.

-------------------- Alternatively -------------------
}
Further, let $[x_1{\mapsto}y_1,\ldots,x_n{\mapsto}y_n]$ denote a mapping from  each $x_i$ to $y_i$ for $i=1..n$ and $f[x\mapsto y]$ denote the function
$f$ updated by (re-)mapping $x$ to $y$.
%

\DONE{Replace partial elements and extended join with default values.}{RM: fixed by updating existing elements.}

\begin{definition}
Let $(\omega_1,\omega_2) \in (M \rightarrow     \packets \to \pow{Q \to \{T,F\}}) \times (E \to \pow{\packets})$ be an abstract element.
Then $\Tr^{\sharp}(\omega_1,\omega_2) \eqdef$
\begin{multline*}
\bigsqcup 
\left\{
\begin{array}{l}
(\omega_1[m{\mapsto}\tilde{\textit{ps}}],\\
\;\,\omega_2[e_i{\mapsto} \omega_2(e_i) \cup \{p_i\}])
\end{array}
\left| \begin{array}{ll}
(1)& m \in M,\\
(2)& p\in \omega_2(e), \inchannel(c,e,m),\\
(3)&\tilde{s} \in \omega_1(m),\ \tilde{p}\in P,\ \\
   &\tilde{s}[p,\tilde{p}]=[p\mapsto \tilde{s}(p), \tilde{p}\mapsto \tilde{s}(\tilde{p})],\\
(4)& \tilde{s}[p,\tilde{p}] \TrMPackSub{(p, c)/(p_i, c_i)_{i=1..k}} \tilde{s}[p,\tilde{p}]',\\
(5)& \tilde{\textit{ps}} = \tilde{s}[\tilde{p}{\mapsto}\{\, \tilde{s}[p,\tilde{p}]'(\tilde{p})\, \}], \\
(6)&\outchannel(c_i,e_i,m), i=1..k
 \end{array}\right.
 \right\}\enspace.
\end{multline*}
\end{definition}

Intuitively, the transformer updates the abstract state by joining
the individual effects obtained by:
(1) considering each middlebox,
(2) considering each input packet to the middlebox,
(3) considering every possible substate for the input packet $p$ and every other packet $\tilde{p}$,
(4) considering every possible substate transition,
(5) adding the new packet state for $\tilde{p}$ to the relevant set,
and
(6) adding each output packet to the corresponding edge.

\ignore{
--------------------------------------------------------------------
}

\begin{proposition}
The running time of $\Tr^{\sharp}$ is $O((|M| + |E|) \cdot |\packets|^2 \cdot 2^{2|Q_{max}|})$, where
$Q_{max}$ denotes the maximal set of queries $Q(m)$ over all middleboxes $m \in M$.
\end{proposition}

\ignore{
We first introduce some notation. Let $m$ be a middlebox, $p$ a packet arriving on channel $e$ and $v \in Q(m) \to \{T,F\}$.
These determine a sequence of actions in $m$, that is traversed when evaluating membership queries according to $v$.
We denote by $O(m,p,e,v)$ the set of packets produced by this sequence of actions.
In order to describe the effect on the middlebox state, we use $\omega \in \packets \to \pow{Q(m) \to \{T,F\}}$ to represent the abstract middlebox state.
We denote by $V(m,p,e,v,\omega)$ the set of pairs $p' \mapsto v'$ that are produced as a result of this sequence of actions
(recall that \verb|r.insert(|$\overline{d}$\verb|)| creates such new pairs).

We first note that the packet state transitions of middleboxes can be equivalently defined as the disjoint union
\TODO{complete} for every $v : Q \to \{T,F\}$.
Let $X = (a, b) \in  (M \times \packets \to \pow{Q \to \{T,F\}}) \times (E \to \pow{\packets})$.
Then applying the transformer on $X$ results in $X' = (a',b')$ where
for every $m \in M$ and $p \in \packets$, $a'(m,p) = a(m,p) \cup $

Let $p \in b(e)$, $v \in a(m,p)$, where $v : Q \to \{T,F\}$. Let $A \subseteq Q \to \{T,F\}$ be the set of all packet
}

\ignore{
\para{Abstract safety verification}
As before, any abstract element in which one of the (sub)-components is $\err$, is an error element.

\begin{definition}[Abstract Safety Verification]
The abstract safety verification problem is to decide whether \linebreak
$\lfp(\Tr^{\sharp})(a_I) = \bigsqcup\limits_{i=1}^\infty {\Tr^{\sharp}}^{i}(a_I) = \err$.
\end{definition}
}


Our algorithm for safety verification computes $\mu^{\sharp} \eqdef \lfp(\Tr^{\sharp})(a_I) = \bigsqcup\limits_{i=1}^\infty {\Tr^{\sharp}}^{i}(a_I)$
and checks whether $\err\in\mu^{\sharp}$.

\para{Complexity of Least Fixpoint Computation}
The height of the abstract domain lattice is determined by the number of packets that can be added to the channels of the network---($|\packets|\cdot|E|$), multiplied by the number of state changes that can occur in any of the middleboxes---$O(|M|\cdot|\packets|\cdot2^{|Q|})$.
The time complexity of the abstract interpretation is bounded by the
height of the abstract domain lattice multiplied by the time complexity of the abstract transformer:
\[
O(|\packets|^4\cdot|E|\cdot|M|\cdot2^{3|Q_{max}|}\cdot(|M| + |E|))\enspace.
\]

\subsection{Soundness and Completeness}


Our algorithm is sound in the sense that it never misses an error state. This follows from the use of a sound abstract interpretation:
\ignore{
The use of a sound abstract interpretation, ensures soundness of our algorithm w.r.t. concrete networks.
In fact, we show soundness w.r.t. the unordered reverting semantics, which implies soundness w.r.t. the concrete (ordered non-reverting) semantics as well, as the collecting semantics of the latter is a subset of the former:
}
\begin{theorem}[Soundness]
\label{thm:soundness-packets}
$\BBsub{\Net}{pa}^{o} \subseteq \BBsub{\Net}{pa}^{ur} \subseteq \gamma(\mu^{\sharp})$. 
\end{theorem}
\ignore{
The above theorem, combined with the bisimulation between the relation state semantics and the packet state semantics, ensures that an error state is never missed by the analysis.
}
%

Our algorithm is also complete relative to the reverting unordered semantics.
\begin{theorem}[Completeness]
\label{thm:completeness-packets}
$\mu^{\sharp} \absorder \alpha({\BBsub{\Net}{pa}^{ur}})$. 
\end{theorem}
The proof of \thmref{completeness-packets} relies on the sticky property formalized by \lemref{sticky-packet-states}.
The theorem states that for reverting unordered networks 
$\mu^{\sharp}$ is at least as precise as applying the abstraction function on the concrete packet state network semantics.
In particular, this implies that if $\mu^{\sharp}$ is an abstract error element then $\err \in {\BBsub{\Net}{pa}^{ur}}$.
As a result, for such networks our algorithm is a decision procedure. For other networks it may produce false alarms, if safety is not maintained by an unordered reverting abstraction.
\ignore{
Due to the bisimulation with the relation state representation, the same holds w.r.t. ${\BB{\Net}^{ur}}$ (the collecting semantics over relation states). 
Thus, for such networks the analysis is precise w.r.t. safety properties.
}

\ignore{
The proof of \thmref{completeness-packets} relies on a finer variant of the sticky states property of reverting unordered networks for the case of packet state representation, which we state next.
\begin{lemma}
\label{lem:sticky-packet-states}
For every channel $e$, packets $p, \tilde{p}$, middlebox $m$ and state $\tilde{v}$ of $\tilde{p}$ in $m$:
If in some reachable configuration channel $e$ contains $p$ and in some (possibly other) reachable configuration the packet state of $\tilde{p}$ in $m$ is $\tilde{v}$, then 
there exists a reachable configuration where simultaneously $e$ contains $p$ and the packet state of $\tilde{p}$ in $m$ is $\tilde{v}$.
\end{lemma}
Intuitively, the proof of the lemma follows as all middleboxes can revert and the unordered semantics enables a scenario where the particular state and packets are reconstructed.
\iflong
\begin{proof}
Follows directly from \lemref{sticky-states} and from the bisimulation between the relation state space and the packet state space of middleboxes.
\end{proof}
\fi
}


\iflonglong
\longlongtrue

\section{Reverting Safety Properties}
\label{sec:RevertRobust}
\ignore{As summarized in \Cref{thm:completeness-packets}, our analysis is precise for unordered reverting networks.}
Recall that we express safety properties via middleboxes in the network.
Therefore, in unordered reverting networks, the possibility to revert applies to the safety property as well.
\iflonglong
As the reverting semantics adds transitions, this
\else
This
\fi
may increase the possible set of transitions of the safety middleboxes, and, in particular, may add transitions into an error state.
\iflonglong
For some temporal safety properties this is a source of imprecision as they cannot be precisely captured by the reverting semantics, thus introducing false alarms.
\else
For some temporal safety properties, this may introduce false alarms.
\fi

\iflonglong
For example, if the safety property forbids a packet from host $h_{\text{ext}}$ to host $h_\text{in}$ before a packet from host $h_\text{in}$ has been sent to $h_\text{ext}$,
then in a reverting network,
even if a packet from host $h_\text{in}$ has been previously sent to $h_\text{ext}$, a revert transition allows the middlebox to return to its initial state, from which a packet from host $h_\text{ext}$ to host $h_\text{in}$ leads to an error state.
\fi

\iflonglong
However, we identify a class of safety middleboxes that is guaranteed not to be a source of imprecision. This class includes any stateless safety middlebox, and in particular isolation middleboxes,
More generally, we provide a sufficient 
condition for a safety property to be precisely expressible in a reverting network.
To do so, we first decouple the enforcement of safety from the forwarding behavior of the network.
For this decoupling, in the sequel we consider safety middleboxes with a single output port that forward any incoming packet (on any input port) to the output port without any modification.
This ensures that safety middleboxes do not affect the forwarding behavior of the network. In particular, the forwarding behavior of safety middleboxes does not depend on their state. The state is only used to enforce safety.
For such safety middleboxes we define:
\else
However, we identify a class of safety middleboxes, which includes isolation middelboxes, that is guaranteed not to be a source of imprecision.
We restrict our attention to safety middleboxes with a single output port that forward any incoming packet (on any input port) to the output port without any modification.
This is not a real restriction; it only means that we decouple enforcement of safety from the forwarding behavior of the network. 
For such safety middleboxes we define:
\fi


\begin{definition}
A safety middlebox $m$ is \emph{revert-robust} if for every sequence of input packets $in = (p_i,c_i)_{i=1..k}$, if no execution of $m$ on $in$, starting from $m$'s initial state, leads to $\err$, then for every suffix $in'$ of $in$,
no execution of $m$ on $in'$ starting from $m$'s initial state leads to $\err$ as well.
\end{definition}


Intuitively, revert-robustness means that the language of ``safe'' sequences of packets is suffix-closed. In particular, any stateless safety middlebox (such as an isolation middleboxes) is revert-robust.
\iflonglong
For example, if the safety middlebox forbids a packet from host $h_\text{ext}$ to host $h_\text{in}$ after a packet from host $h_\text{in}$ has been sent to $h_\text{ext}$, then it is revert-robust. The reason is that, in this example, the ``safe'' input sequences are ones where no packet from host $h_\text{ext}$ to host $h_\text{in}$ has a preceding packet from host $h_\text{in}$ to $h_\text{ext}$. Therefore any suffix of a safe input sequence is also safe. As a result, such a safety middlebox will not introduce false alarms in a reverting network, as reverting transitions will just make the middlebox ``forget'' the prefix of the sequence. (Note that it will also not make the network wrongfully safe, as safety requires that all executions, including the ones that do not use revert transitions, are safe.)
\fi
\iflonglong
Next, we claim that revert-robustness is a sufficient 
condition for not losing precision of the analysis (i.e., not introducing false alarms) due to the revert transitions of the safety middlebox.
In order to formalize this claim, we need the following definitions. For a network $\Net$ with a set of middleboxes $M$, a subset $S \subseteq M$, and a semantic identifier $i \in \{o,u,or,ur\}$, we denote by ${\BBsub{\Net}{pa}^{i \setminus S}}$ the corresponding network collecting semantics, with the exception that no reverting transitions are applied to the middleboxes in $S$ (when applicable).
We then have:
\else
Next, we claim that revert-robustness suffices to ensure that the safety middleboxes themselves are not a source of false alarms in our analysis.
\fi
\begin{lemma}
\label{lem:revert-robust}
Let $\Net$ be a network such that all of its safety middleboxes, $S \subseteq M$, are revert-robust.
Then for every $i \in \{o,u,or,ur\}$, $\err \in {\BBsub{\Net}{pa}^{i \setminus S}}$ if and only if $\err \in {\BBsub{\Net}{pa}^{i}}$,
where ${\BBsub{\Net}{pa}^{i \setminus S}}$ is the same as $ {\BBsub{\Net}{pa}^{i}}$, except that no reverting transitions are applied to the middleboxes in $S$.
\end{lemma}
\iflonglong
This means that the network is safe (under any of the semantics) if and only if it is safe with the same semantics except that all safety middleboxes are non-reverting.
\fi

\iflonglong
\begin{proof}
The direction from left to right is trivial, as the reverting semantics is a sound approximation, hence a computation leading to error when $S$ is non-reverting also exists when $S$ is reverting.
In order to prove the converse direction we denote by $\Net$ the network where all middleboxes including $S$ may revert and by $\Net '$  the network where $S$ may not revert. We prove that if all the computations of $\Net '$ are safe then so are the computations of $\Net$.
The proof is straightforward. 
We observe that for every scenario $s$ in $\Net$ there is a corresponding scenario in $\Net '$ which is identical to $s$ other than the behavior of the safety middleboxes (this is because safety middleboxes do not affect forwarding of packets).
Consider a safety middlebox $m$ and an arbitrary step $i$ in the scenario.
Let $p_1,\dots,p_\ell$ be the sequence of packets that $m$ processed until step $i$ and let $p_r,\dots,p_\ell$ be the packets it processed since it was last reverted.
Since $\Net '$ is safe, it follows that in $\Net '$ the middlebox $m$ is not in $\err$.
As $m$ is revert-robust and $p_r,\dots,p_\ell$ is a suffix of $p_1,\dots,p_\ell$, then $m$ is also not in $\err$ state in $\Net$.
Thus, we get that for every $s, i$ and $m$, the middlebox $m$ is not in $\err$ state.
Hence, $\Net$ is safe and the proof is complete.
\end{proof}
\fi

\ignore{
\begin{lemma}
For every non revert-robust safety middlebox $m$ there exists a network $\Net$ with one \emph{non-reverting} forwarding middlebox $f$, one safety middlebox $m$ and a single host such that $\Net$ is safe only if $m$ is not reverting.
\end{lemma}
\begin{proof}
As we assume that $m$ is not revert-robust, then there exist a sequence of packets $p_1,\dots,p_\ell$ and a suffix $p_r,\dots,p_\ell$ such that $m$ reaches $\err$ over $p_r,\dots,p_\ell$ and not over $p_1,\dots,p_\ell$.
Our witness $\Net$ is a network where a middlebox $f$ produces the sequence $p_1,\dots,p_\ell$ to $m$ and then never transmit packets again.
To maintain the needed order of packets we construct the network such that all of $m$ output packets arrive to $f$, and we construct $f$ such that a single packet is outputted in response to a packet from $m$.
It is straight forward to see that if $m$ is non-reverting, then $\err$ is never reached.
On the other hand, if $m$ reverts after $r$ packets, then $\err$ state is reached after $p_\ell$ is processed.
\end{proof}
}

\ignore{
The next lemma states that revert-robustness is a sufficient but also necessary condition for not losing precision of the analysis (i.e., not introducing false alarms) due to the revert transitions of the safety middlebox.
In order to formalize this claim we need the following definitions. For a network $\Net$ that includes middlebox $m$ we denote by ${\BBsub{\Net}{pa}^{ur \setminus m}}$ the unordered reverting network collecting semantics, with the exception that no reverting transitions are added to $m$.

\begin{lemma}
\TODO{need to formalize the "necessary and sufficient" claim. I wrote something but not sure}
Let $m$ be a safety middlebox. The condition that for any network $\Net$ where $m$ is the only safety middlebox, $\err \in {\BBsub{\Net}{pa}^{ur}}$ if and only if $\err \in {\BBsub{\Net}{pa}^{ur \setminus m}}$ is equivalent to the condition that $m$ is revert-robust.
\end{lemma}

}

%

\else
\para{Properties}
Recall that we express safety properties via middleboxes in the network.
Therefore, in unordered reverting networks, the possibility to revert applies to
the safety property as well, and may introduce false alarms due to addition of
behaviors leading to error. However, for safety properties such as isolation
which are suffix-closed (i.e., all the suffixes of a safe run are themselves
safe runs), this cannot happen (\Cref{sec:RevertRobust}).
\fi


\section{Implementation and Initial Evaluation}
\label{sec:empirical}

In this section, we describe our implementation of the analysis described in \secref{abs-int}, and report our initial experience running the algorithm on a few example networks.

\para{Implementation}
We have developed a compiler, \verb|amdlc|, which takes as input a network topology \cbstart and its initial state \cbend (given in \verb|json| format) and \AMDL\ programs for the middleboxes that appear in the topology. The compiler outputs a Datalog program, which can then
be efficiently solved by a Datalog solver.
%
Specifically, we use LogicBlox~\cite{aref2015design}.


The generated Datalog programs include three relations: (i)~\verb|packetsSeen|,
which stores the packets sent over the network channels;
(ii)~\verb|middleboxState|, which stores the packet state of individual packets
in each middlebox (i.e., the possible valuation of each middlebox program's
queries for each individual packet); and (iii)~\verb|abort|, which stores the
middleboxes that have reached an $\err$ state.

We encode the packets that hosts can send to their neighboring middleboxes and
the initial state of the middleboxes as Datalog \emph{facts} (edb), and the
effects of the middlebox programs, i.e. relation update actions and packet
output actions, as Datalog \emph{rules} (idb).

We then use the datalog engine to compute the fixed point of the datalog
program. That fixed point is exactly the least fixed point  $\mu^{\sharp} \eqdef
\lfp(\Tr^{\sharp})(a_I) = \bigsqcup\limits_{i=1}^\infty {\Tr^{\sharp}}^{i}(a_I)$




\para{Evaluation}
The main challenge in acquiring realistic benchmarks is that middlebox
configuration and network topology are considered security sensitive, and as a
result enterprises and network operators do not release this information to the
public. Consequently, we benchmarked our tool using the synthetic topologies and
configurations described by \cite{DBLP:conf/nsdi/PandaLASS17}.

Our benchmarks focus on datacenter networks and enterprise networks. The set of
middleboxes we used in our datacenter benchmarks is based on information
provided in \cite{potharaju2013demystifying}, and on conversations with
datacenter providers. We ran both a simple case where each tenant machine is
protected by firewalls and an IPS (Intrusion Prevention System); and a more
complex case where we use redundant servers and distribute traffic across them
using a load balancer. Our enterprise topology is based on the standard topology
used in a variety of university departments including UIUC (reported in
\cite{MaiKACGK11}), UC Berkeley, Stanford, etc. which employ firewalls and an IP
gateway.

%

We ran two scaling experiments, measuring how well our system scales when the
number of hosts or the number of middleboxes in the network increases
The experiments were run on Amazon EC2 r4.16 instances with 64-core CPUs
and 488GiB RAM.

\ignore{
\para{Middlebox Chain}
We tested our implementation on the networks illustrated in \figref{example:good} and \figref{example:bad}. Our system found the configuration error in \figref{example:bad}, and successfully proved that the network in \figref{example:good} is safe. 


\para{Session Firewalls}
We tested our implementation on the network illustrated in the running example (\figref{running-ex-topo}). Our system successfully proved that the network is safe.
}

\begin{figure}[t]
    \centering
    \includegraphics[width=0.5\textwidth]{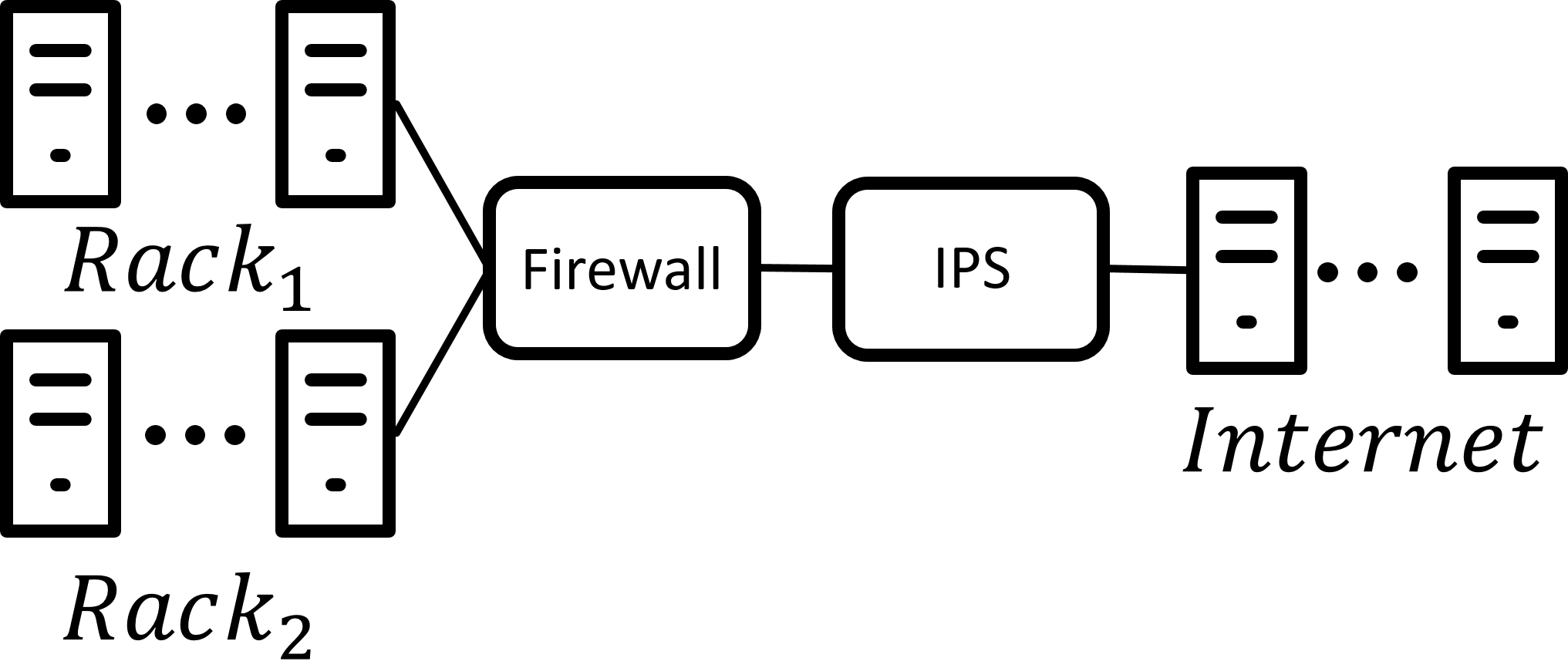}
    \caption{Topology of the datacenter example.}
    \label{fig:multi-tenant-datacenter}
\end{figure}

\para{Multi Tenant Datacenter Network}
\figref{multi-tenant-datacenter} illustrates the topology of a multi tenant
datacenter. Each rack hosts a different tenant, and the safety property we wish
to verify is isolation between the hosts of the two racks. In this example the
network also employs an IPS to prevent malicious traffic from reaching the
datacenter. Actual IPS code is too complex to be accurately modeled in \AMDL;
instead we over-approximate the behaviour of an IPS by modeling it as a process
that non-deterministically drops incoming packets.


\para{Enterprise Network}
\figref{enterprise} illustrates the topology of an enterprise network.
The enterprise network consists of three subnets, each with a different security policy.
The \emph{public} subnet is allowed unrestricted access with the outside network.
The \emph{quarantined} subnet is not allowed any communication with the outside network.
The \emph{private} subnet can initiate communication with a host in the outside network, but hosts in the outside network cannot initiate communication with the hosts in the \emph{private} subnet.

To evaluate the feasibility of our solution, we ran the analysis of \figref{enterprise} on networks with varying numbers of hosts ranging from 20 to 2,000. Our implementation successfully verified a network with 2,000 hosts in under four hours, suggesting that the implementation could be used to verify realistic networks.
%
%
\figref{enterprise-scaling} shows the times of the analysis on an enterprise network with 20--2,000 hosts.

\para{Datacenter Middlebox Pipeline}
\Cref{fig:mbox-scaling-topo} describes a datacenter topology with a pipeline of middleboxes connecting servers to the Internet. The topology contains multiple middlebox pipelines for load-balancing purposes and to ensure resiliency. We use this topology to test the scalability of our approach w.r.t the size of the network, by adding additional middlebox pipelines and keeping the number of hosts constant.

\figref{mbox-scaling-test} shows the running times of the analysis of a datacenter with 3--189 middleboxes (1--32 middlebox chains). All topologies contained 1000 hosts.

\begin{figure}[t]

\begin{subfigure}[t]{0.5\textwidth}
    \centering
    \includegraphics[width=\textwidth]{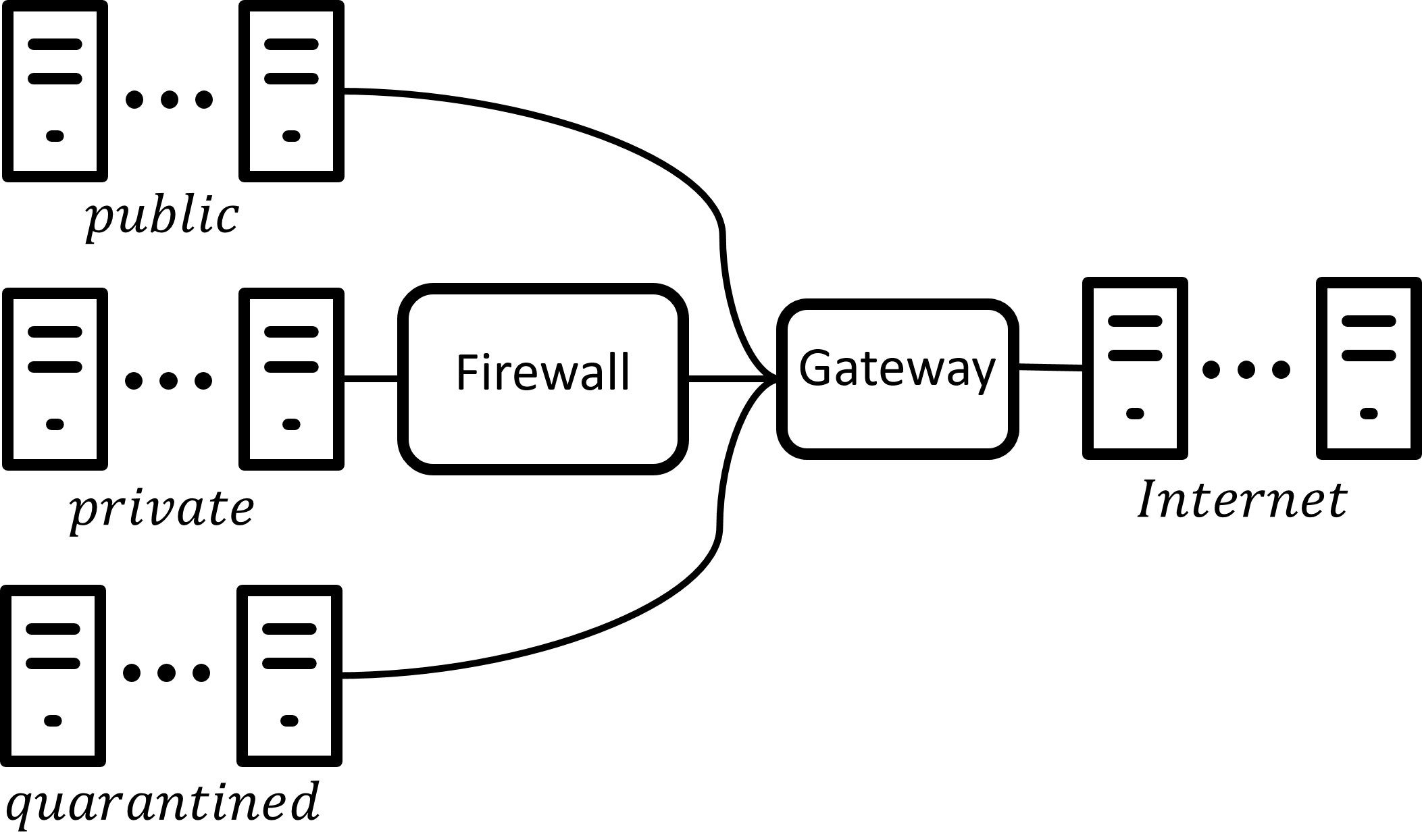}
    \caption{Enterprise}
    \label{fig:enterprise}
\end{subfigure}
~
\begin{subfigure}[t]{0.5\textwidth}
  \centering
  \includegraphics[width=\textwidth]{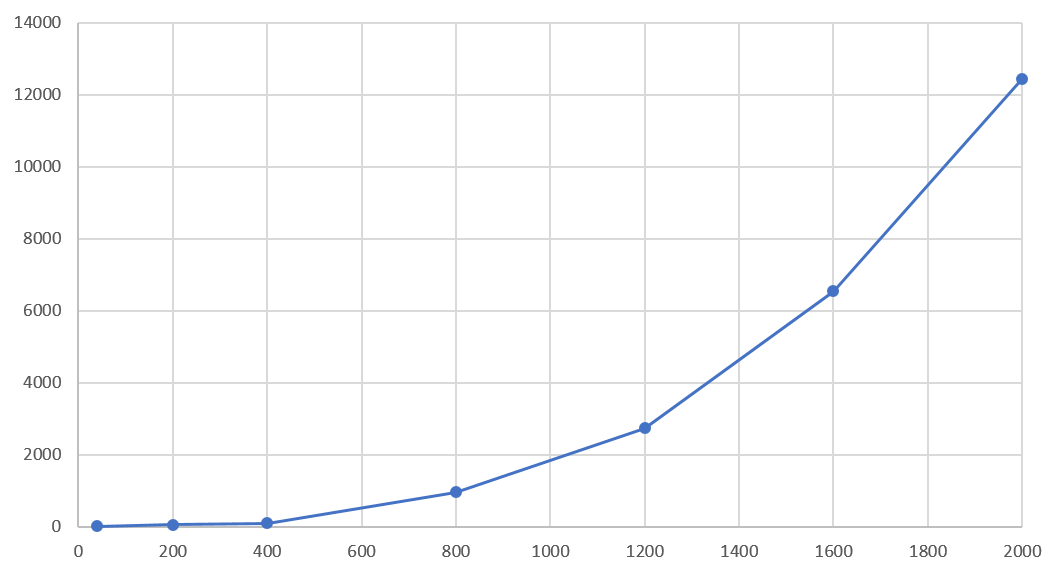}

  \caption{Running time (seconds).}
  \label{fig:enterprise-scaling}
\end{subfigure}
\caption{Topology and running times of the host scalability test.}
\end{figure}

\begin{figure}
\begin{tabular}{c}
\begin{subfigure}[t]{0.45\textwidth}
    \centering
    \includegraphics[width=\textwidth]{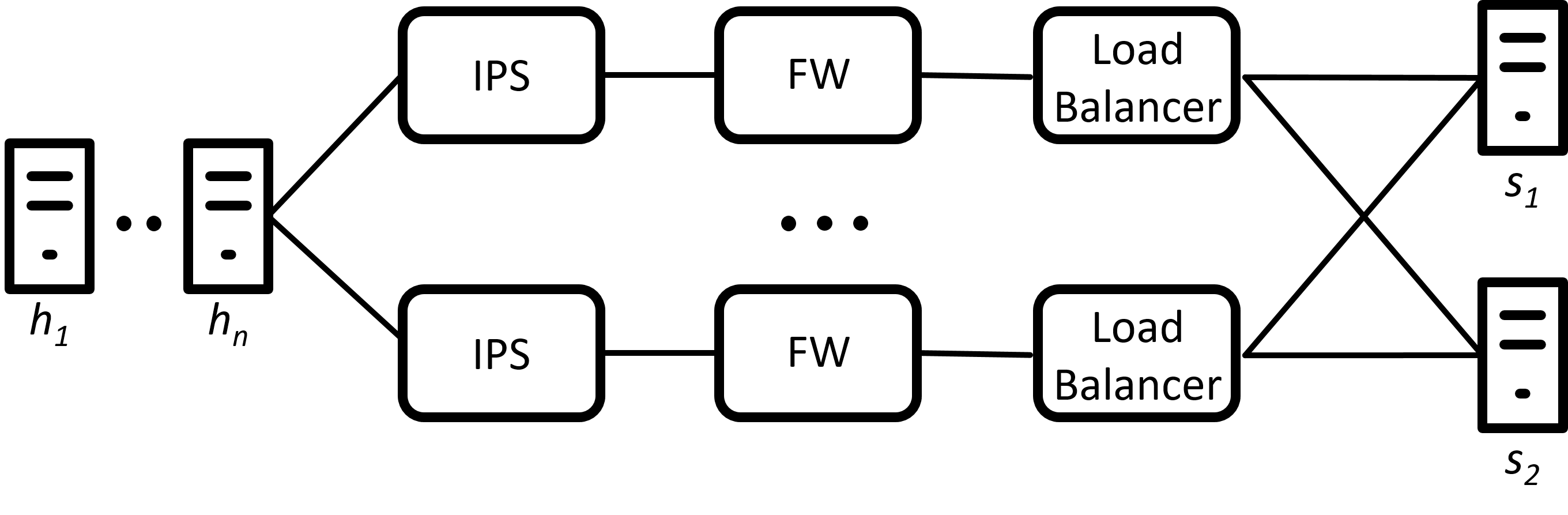}
    \caption{Topology with multiple middlebox-pipelines}
    \label{fig:mbox-scaling-topo}
\end{subfigure}
\end{tabular}
~\hspace*{.2in}
\begin{subfigure}[t]{0.45\textwidth}
  \centering
  \raisebox{-0.5\height}{\includegraphics[width=\textwidth]{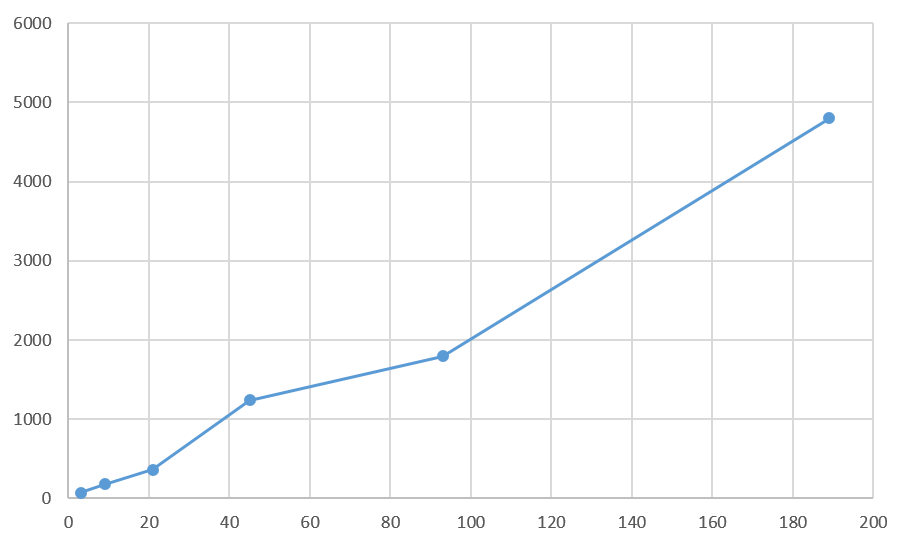}}
  \caption{Running time (seconds).}
  \label{fig:mbox-scaling-test}
\end{subfigure}
\caption{Topology and running times of the network topology scalability test.}
\end{figure}


\section{Concluding Remarks and Related Work}
In this paper, we applied abstract interpretation for efficient verification of networks with stateful nodes.
We now briefly survey closely related works in this area.

\ignore{
\para{Safety checking of general distributed systems} A general line of research looks at
automatically verifying general distributed programs~\cite{DBLP:conf/popl/JhalaM07}.
In this work it is natural to model programs as Petri Nets which can then be combined
with CFA abstractions~\cite{DBLP:conf/sas/DOsualdoKO13}.
\AMDL\ supports a limited programming model, and
as a result our tool is efficient and can infer invariants in polynomial time.
Investigating how \AMDL\ can be generalized while continuing to provide polynomial time verification is left to future
work.
}

\para{Topology Independent Network Verification} Early work in network verification focused on proving
correctness of network
protocols~\cite{clarke1998using,ritchey2000using}. Subsequent work
in the context of software define networking (SDN) including
Flowlog~\cite{FlowLog} and VeriCon~\cite{VeriCon14} looked at
verifying the correctness of network applications (implemented as
middleboxes or in network controllers) independent of the topology
and configuration of the network where these were used. However,
since this problem is undecidable, these methods use bounded model
checking or user provided inductive invariants, which are hard to
specify even in simple network topologies.

\para{Verifying Immutable Network Configurations}
Verifying networks with immutable states is an active line of research~\cite{MaiKACGK11,nsdi:KCZCMW13,CCR:KhurshidZCG12,nsdi:CaniniVPKR12,nsdi:KVM12,Verificare,FMACAD:SNM13,anderson2014netkat,NetKat15}.
In the future, we hope to combine our abstraction with the techniques used in these papers.
We hope to use similar techniques to Veriflow~\cite{CCR:KhurshidZCG12} to handle switches
more efficiently, and leverage compact header representation described in NetKat~\cite{NetKat15}.

\para{Stateful Network Verification}
Previous works provide useful tools for detecting errors in firewalls~\cite{mayer2000fang,marmorstein2005tool,nelson2010margrave}.
Buzz~\cite{Fayaz2016BuzzTC} and SymNet~\cite{Stoenescu2016SymNetSS} have looked at how to use symbolic execution and
packet generation for testing and verifying the behavior of stateful networks. These works implement testing techniques
rather than verifying network behavior and are hence complementary to our approach.

Velner et al.~\cite{velner2016some} show that checking safety in stateful
networks is undecidable, necessitating the use of overapproximations. They
provide a general algorithm for checking safety using Petri nets. This algorithm
has high complexity and scales poorly. They also provide an efficient algorithm
for checking safety in a limited class of networks.




\para{Exploring Network Symmetry}
Recent work explored the use of bisimulation to leverage the extensive symmetry
found in real network topologies~\cite{DBLP:conf/vmcai/NamjoshiT13} to
accelerate stateless~\cite{DBLP:conf/popl/PlotkinBLRV16} and
stateful~\cite{DBLP:conf/nsdi/PandaLASS17} network verification. Both approaches
are not automatic. We are encouraged by the fact that our automatic approach
achieves performance comparable to VMN~\cite{DBLP:conf/nsdi/PandaLASS17} on the
same examples without requiring human intervention. We attribute this
improvement to modularity and to the use of packet state representation.




\para{Extensible Semantics}
Previous works have explored ideas similar to the reverting semantics, to obtain
complexity and decidability results in different settings.

In \cite{esparza2013parameterized} the authors analyze the complexity of
verifying asynchronous shared-memory systems. They use \emph{copycat} processes
that mirror the behaviour of another process to show that executions are
extensible, similarly to how our work uses the sticky packet states property
(\Cref{lem:sticky-packet-states}). In their model, when the processes are finite
state machines, they obtain coNP-complete complexity for verification.

In \cite{finkel2001well} the authors explore a more general setting of
well-structured transition system, and present the \emph{home-state idea}, which
allows the system to return to its initial state (essentially, revert). They
obtain decidability results for well-structured transition systems with a
home-state, but do not show any tighter complexity results. 
 \paragraph{Acknowledgments}

    We thank our anonymous shepherd, and anonymous referees for insightful
    comments which improved this paper.
    We thank LogicBlox for providing us with an academic license for their
    software, and Todd J. Green and Martin Bravenboer for providing technical
    support and helping with optimization.
    This publication is part of projects that have received funding from the
    European Research Council (ERC) under the European Union's Seventh Framework
    Program (FP7/2007--2013) / ERC grant agreement no. [321174-VSSC], and Horizon
    2020 research and innovation programme (grant agreement No [759102-SVIS]).
    The research was supported in part by Len Blavatnik and the Blavatnik Family
    foundation, the Blavatnik Interdisciplinary Cyber Research Center, Tel Aviv
    University, and the Pazy Foundation.
    This material is based upon work supported by the United States-Israel
    Binational Science Foundation (BSF) grants No. 2016260 and 2012259.
    %
    %
    This research was also supported in part by NSF grants 1704941 and 1420064, and
    funding provided by Intel Corporation.

\bibliographystyle{abbrv}
\bibliography{refs}

\iflong
\clearpage
\appendix

\section{Proofs}
\label{sec:proofs}

In this section, we include proofs for some of the key claims made in the paper.

\begin{proof}[Proof of \thmref{reverting-unordered-conphard} (Undecidability)]
It is well known that an automaton with an ordered channel of messages (also
known as a \emph{channel machine}) can simulate a Turing machine. The channel
can trivially store the content of a Turing machine tape, and the automaton can
simulate the transitions of the machine. This can be used to easily show that in
the absence of reverting the isolation problem over ordered channels is
undecidable even when there is only one host, and one middlebox with a self
loop.

When reverting is possible, we add auxiliary packet type and middlebox states.
Whenever in initial state, the middlebox sends a special packet over its self
loop, and discards all arrived packets until it receives the special
packet~\footnote{Note that for this step it is crucial that the channels are
FIFO.}. This empties the self loop from its content, which intuitively, resets
the tape of the Turing machine. Hence, when the middlebox reverts, so does the
Turing machine. Thus, the isolation property is violated if and only if the
Turing machine reaches an accepting state, and the undecidability proof follows.
\end{proof}

\begin{proof}[Proof of \thmref{reverting-unordered-conphard} (coNP-hardness)]
We prove that if the number of queries in a middlebox is not a constant (i.e.,
it depends on other parameters of the problem), then the safety problem is coNP-
hard even when the network consists of only one middlebox and one host.
The proof is by reduction from the Boolean unsatisfiability problem of
propositional formulas.

Given a formula $\phi$ with $n$ variables $x_1,\dots,x_n$ we construct a network
with one host and one middlebox $m$, such that $m$ has only one port, connected
to $h$. The packet types are $x_1,\neg x_1, \dots, x_n, \neg x_n$, i.e., there
are $2n$ packet types, one for each literal. The middlebox has two nullary
relations, $O_i$ and $V_i$, for every $i\in\{1,\dots,n\}$, where intuitively,
$O_i$ indicates whether a packet of type $x_i$ or $\neg x_i$ already occurred
and $V_i$ indicates if the first such packet is positive ($x_i$) or negative
($\neg x_i$). That is, the $O_i$ relations indicate which variables are
assigned, while the $V_i$ relations store the assignment. Initially all the
relations are initialized to \False\ (i.e., no variable is assigned). Upon
receiving a packet of type $x_i$ or $\neg x_i$, the middlebox updates the
relation $V_i$ only if $O_i$ is \False, in which case $O_i$ is also updated to
\True. If the packet type is $x_i$, then $V_i$ is updated to \True. Otherwise it
is updated to \False.
In addition, whenever the interpretation of $O_i$ and $V_i$ satisfies $\phi$, 
the middlebox aborts. Clearly, the size of the code of $m$ is polynomial and
safety is violated if and only if $\phi$ is satisfiable. We note that possible
resets do not affect the safety of the network.
\end{proof}

\begin{lemma}[Sticky Packets Property]\label{lem:sticky-packets}
For every channel $e$ and packet $p$: If in some reachable configuration $e$
contains $p$, then every run can be extended such that $e$ will eventually
contain $p$. Moreover, every run can be extended such that $e$ will eventually
contain $n$ copies of $p$ (for every $n>0$).
\end{lemma}

\begin{proof}[Proof of \lemref{sticky-packets}]
The proof relies on the reverting property and on the fact that the channels are
unordered.

Let $\sigma_0$ be a reachable configuration in which $p$ occurs in $e$, and let
$s_0$ be the scenario that led to it, i.e., the sequence of events that took
place. Consider an arbitrary run (scenario) $\pi$. One can extend $\pi$ with the
following scenario: First all the middleboxes return to their initial state.
Second, scenario $s_0$ occur, i.e., only packets from scenario $s_0$ are
processed, and the other packets are ignored. This extension is possible because
the channels are unordered.

To construct a scenario in which $e$ contains $n$ copies of $p$, we just
concatenate the above mentioned extension $n$ time.
\end{proof}

\begin{lemma}[Sticky States Property]
\label{lem:sticky-states}
For every channel $e$, packet $p$, middlebox $m$ and state $s$ of $m$: If, in
some reachable configuration, channel $e$ contains $p$ and in some (possibly
other) reachable configuration $m$ is in state $s$, then there exists a
reachable configuration where simultaneously $e$ contains $p$ and $m$ is in
state $s$.
\end{lemma}

\begin{proof}[Proof of \lemref{sticky-states}]
Let $(p_1,\dots,p_\ell)$ be the sequence of packets that $m$ processed from the
latest reset until it arrives to state $\sigma_m$ in the given witness scenario.

Consider an arbitrary run. By \lemref{sticky-packets} we can extend this run
such that $p_1,\dots,p_\ell$ are pending packets in the ingress channel of
middlebox $m$ and $p$ is pending in $e$ (if some of the packets occur more than
once in the sequence, then by the same lemma we may assume that there are
multiple copies of those packets).

We further extend the run with a reset event for middlebox $m$. Finally, we
extend the scenario such that in the next $\ell$ steps $m$ will process
$p_1,\dots,p_\ell$ reaching state $\sigma_m$.
\end{proof}

\begin{proof}[Proof of \lemref{sticky-packet-states}]
The proof follows directly from \Cref{lem:sticky-packets,lem:sticky-states,lem:bisimulation}.
\end{proof}

\begin{proof}[Proof of \thmref{completeness-packets}]
In order to prove completeness it is enough to show that every application of
the best abstract transformer results in an abstract value that is less or equal
than the result of applying the abstraction function on the concrete least fixed
point (i.e., the reachable states of the network w.r.t unordered reverting
packet state space semantic). The proof is by induction over $n$, the number of
times we apply the transformer. The proof for $n=0$ is trivial. For $n>1$, let
$p,\tilde{p}$ and $m$ be packets and a middlebox. By the induction hypothesis
for every packet state $v \in \omega_1(m)(\tilde{p})$ there is a concrete
reachable middlebox state such that the state of $m$ over packet $\tilde{p}$ is
$v'$
and for every packet $p\in \omega_2'(e)$ there is a reachable concrete
configuration where $p$ is in $e$. Hence, by~\lemref{sticky-packet-states},
there exists a concrete reachable configuration in which $p$ is in $e$ and the
state of $m$ over packet $\tilde{p}$ is $v$. Therefore, by definition of
$\omega_1'$ and $\omega_2'$, every new state in $\omega_1'(m)(\tilde{p})
\setminus \omega_1(m)(\tilde{p})$ has a corresponding concrete reachable state,
and likewise for any new pending packet in $\omega_2'(e)\setminus \omega(2)$.
The proof is complete.
\end{proof}

\begin{proof}[proof of \lemref{revert-robust}]
The direction from left to right is trivial, as the reverting semantics is a
sound approximation, hence a computation leading to error when $S$ is non-
reverting also exists when $S$ is reverting. In order to prove the converse
direction we assume that $\err \not \in{\BBsub{\Net}{pa}^{i \setminus S}}$ and
prove that all the computations of ${\BBsub{\Net}{pa}^{i}}$ are safe. The proof
is straightforward. We observe that for every computation $s$ in
${\BBsub{\Net}{pa}^{i}}$ there is a corresponding computation in
${\BBsub{\Net}{pa}^{i \setminus S}}$ which is identical to $s$ other than the
behavior of the safety middleboxes (this is because safety middleboxes do not
affect forwarding of packets). Consider a safety middlebox $m$ and an arbitrary
step $k$ in the computation. Let $p_1,\dots,p_\ell$ be the sequence of packets
that $m$ processed until step $i$ and let $p_r,\dots,p_\ell$ be the packets it
processed since it last reverted. Since $\err \not \in{\BBsub{\Net}{pa}^{i
\setminus S}}$ it follows that in particular the middlebox $m$ is not in $\err$
state. As $m$ is revert-robust and $p_r,\dots,p_\ell$ is a suffix of
$p_1,\dots,p_\ell$, then $m$ is also not in $\err$ state in
${\BBsub{\Net}{pa}^{i}}$ (where it may revert). Thus, we get that for every $s,
k$ and $m$, the middlebox $m$ is not in $\err$ state. Hence, $\err \not \in
{\BBsub{\Net}{pa}^{i}}$ and the proof is completed.
\end{proof}

\section{The Semantics of \AMDL}
\label{sec:amdl-semantics}

\newcommand{\ChannelOutput}{\textit{send}}
\newcommand{\ch}{\textit{ch}}
\newcommand{\ac}{\textit{ac}}
\newcommand{\relderive}{\longrightarrow_{\textsf{R}}}
\newcommand{\releval}[1]{\textsf{R}\lsyn #1 \rsyn}
\newcommand{\packderive}{\longrightarrow_{\textsf{P}}}
\newcommand{\packeval}[1]{\textsf{P}\lsyn #1 \rsyn}

In this section, we define two semantics for middleboxes---the one based on relation
states and the one based packet states. We then prove that both semantics
are bisimilar.

\para{A Note on Field Binding.}
A $\pblock$ construct binds the atoms in a packet received on a channel to field names before executing
a guarded commands.
We will assume that there is at most one $\pblock$ construct per incoming channel.
This assumption does not impose a restriction, since two $\pblock$ constructs
$\ch ~ \textbf{?}~ (f_1,\ldots,f_k) ~ \Rightarrow~ \gc_1$
and
$\ch ~ \textbf{?}~ (g_1,\ldots,g_k) ~ \Rightarrow~ \gc_2$
over the same channel $\ch$ can be
automatically merged into a single $\pblock$ construct via the source-to-source transformation
\[
\ch ~ \textbf{?}~ (f_1,\ldots,f_k) ~ \Rightarrow~ \textbf{if}~ \gc_1 \Box \gc_2[f_1/g_1,\ldots,f_k/g_k]~ \textbf{fi}
\]
where the field names of the second $\pblock$ construct are substituted appropriately for the field names of the first $\pblock$ construct. (Technically, the transformation first extends the sequence of atoms of the $\pblock$
construct with fewer number of atoms by adding dummy atoms.)
This assumption allows us to access the atom $a_i$ of the incoming packet by indexing into the
sequence of fields, as $f_i$.

\subsection{Relation State Semantics}
\label{sec:amdl-semanics-relations}

We start by defining a big-step semantics for relation states.

Let $\mboxe$ be a fixed middlebox.

For simplicity of the presentation, we consider the case where $P \eqdef (H \times H \times T)$ denotes the set of all packets. (The adaptation to other definitions of the packets space is straightforward.)
Let $C_m$ denote the set of channels of $m$. We define the sequence of pairs of packets and channels to be sent following a transition of the middlebox $m$ on every channel as $\Cont \eqdef (P \times C_m)^{*}$.
The semantics of guarded commands, actions, conditions, and atoms is given
in the context of a middlebox state $s \in \Sigma[\mboxe] = \rels(\mboxe) \to \pow{D(m)}$ and a
packet $p$.

We start by defining in \figref{amdl-semantics-1} semantic evaluation functions for atoms and conditions:
\[
\begin{array}{lcl}
\releval{\cdot} : \gV{\atom}      & \rightarrow & P \rightarrow (T \cup H)\\
\releval{\cdot} : \gV{\cond}      & \rightarrow & (\Sigma[\mboxe] \times P) \rightarrow \{\True,\False\}
\end{array}
\]

\begin{figure*}
\[
\begin{array}{|c|}
\hline
\begin{array}{ll}
\releval{f_i} p \eqdef a_i & p=(a_1,\ldots,a_k) \\
\releval{(f_{j_1},\ldots,f_{j_k})} p \eqdef (a_{j_1},\ldots,a_{j_k}) & p=(a_1,\ldots,a_k) \\
\releval{h} p \eqdef h & h \in H \\
\releval{t} p \eqdef t & t \in T \\
\end{array}\\
\hline
\begin{array}{rcl}
\releval{\textbf{true}} (s,p) &\eqdef& \True\\
\releval{\textbf{false}} (s,p) &\eqdef& \False\\
\releval{c_1\;\band\;c_2} (s,p) &\eqdef& \left\{
                                           \begin{array}{ll}
                                             \True, & \releval{c_1}(s,p)=\True \text{ and } \releval{c_2}(s,p) = \True\hbox{;} \\
                                             \False, & \hbox{otherwise.}
                                           \end{array}
                                         \right.
\\
\releval{\bnot~ c} (s,p) &\eqdef& \left\{
                                           \begin{array}{ll}
                                             \False, & \releval{c} (s,p)=\True\hbox{;} \\
                                             \True, & \hbox{otherwise.}
                                           \end{array}
                                         \right.
\\
\releval{a_1 = a_2} (s,p) &\eqdef& \left\{
                                           \begin{array}{ll}
                                             \True, & \releval{a_1} p = \releval{a_2} p\hbox{;} \\
                                             \False, & \hbox{otherwise.}
                                           \end{array}
                                         \right.
\\
\releval{\overline{a}~ \bin ~ r} (s,p) &\eqdef& \left\{
                                           \begin{array}{ll}
                                             \False, & s=\err\hbox{;} \\
                                             \True,  & \releval{\overline{a}} p \in s(r)\hbox{;} \\
                                             \False, & \hbox{otherwise.}
                                           \end{array}
                                         \right.
\\
\end{array}\\
\hline
\end{array}
\]
\caption{\label{fig:amdl-semantics-1}Semantic evaluation of atoms and conditions.}
\end{figure*}

\figref{amdl-semantics-2} defines transition relations for guarded commands, blocks, and middleboxes:
\[
\begin{array}{c}
\releval{\cdot} : \gV{\action} \rightarrow (\Sigma[\mboxe] \times P \times \Cont) \times (\Sigma[\mboxe] \times P \times \Cont)\\
\releval{ \cdot } : \gV{\gc}  \rightarrow (\Sigma[\mboxe] \times P \times \Cont) \times (\Sigma[\mboxe] \times P \times \Cont)\\
\releval{ \cdot } : \gV{\pblock} \rightarrow (\Sigma[\mboxe] \times (P \times C_m)) \times (\Sigma[\mboxe] \times \Cont)\\
\releval{ \cdot } : \gV{\textit{mbox}} \rightarrow (\Sigma[\mboxe] \times (P \times C_m)) \times (\Sigma[\mboxe] \times \Cont) \enspace.\\
\end{array}
\]

A guarded command accepts a middlebox state, an assignment of fields to values, and a mapping from
output channels to their output content (i.e., the sequences of packets that should be delivered to them). It returns the updated state, the (same) assignment of fields to values, and the new mapping from channels to content.

A block accepts a middlebox state and a packet on a specified input channel
and returns the updated state and the output sent to the output channels.
A middlebox non-deterministically chooses between its blocks.

\begin{figure*}
\[
\begin{array}{|c|}
\hline
\\
\begin{array}{lll}
\langle \ch ~ !~ \overline{a}, (s,p,\ChannelOutput)\rangle & \relderive (s,p,\ChannelOutput) & s=\err\\
\langle \ch ~ !~ \overline{a}, (s,p,\ChannelOutput)\rangle & \relderive (s,p,\ChannelOutput \cdot (\releval{\overline{a} } p,\ch)) & s\neq\err\\
\langle r(\overline{a})~ \textbf{:=} ~c, (s,p,\ChannelOutput) \rangle & \relderive (s,p,\ChannelOutput) & s=\err\\
\langle r(\overline{a})~ \textbf{:=} ~c, (s,p,\ChannelOutput) \rangle & \relderive (s[r\mapsto s(r) \cup \{\releval{\overline{a}} p\}],p,\ChannelOutput) & \releval{c} (s,p)=\True\\
\langle r(\overline{a})~ \textbf{:=} ~c, (s,p,\ChannelOutput) \rangle & \relderive (s[r\mapsto s(r) \setminus \{\releval{\overline{a}} p\}],p,\ChannelOutput) & \releval{c} (s,p)=\False\\
\langle \textbf{abort}, (s,p,\ChannelOutput)\rangle & \relderive (\err,p,\ChannelOutput) &\\
\end{array}\\
\\
\inference{\langle ac_1, (s,p,\ChannelOutput)\rangle \relderive (s',p,\ChannelOutput') & \langle ac_2, (s',p,\ChannelOutput')\rangle \relderive (s'',p,\ChannelOutput'')}
{\langle ac_1; ac_2, (s,p,\ChannelOutput)\rangle \relderive (s'',p,\ChannelOutput'')}\\
\\
\hline
\\
\begin{array}{c}
\begin{array}{rcll}
\langle c \Rightarrow \ac, (s,p,\ChannelOutput)\rangle &\relderive& \releval{ \ac} (s,p,\ChannelOutput) & \text{if } \releval{c} (s,p)=\True\\
\langle c \Rightarrow \ac, (s,p,\ChannelOutput)\rangle &\relderive& (s,p,\ChannelOutput) & \text{if } s=\err
\end{array}\\
\\
\inference{\langle g_i, (s,p,\ChannelOutput)\rangle \relderive (s',p,\ChannelOutput')}
{\langle\textbf{if}~ g_1 \Box \ldots \Box g_n~ \textbf{fi}, (s,p,\ChannelOutput)\rangle \relderive (s',p,\ChannelOutput')} \;\;\; i \in \{1,\ldots,n\}\\
\end{array}\\
\\
\hline
\\
\inference{\langle g, (s,p,\emptyset) \rangle \relderive (s',p,\ChannelOutput)}
{\langle \ch ~ \textbf{?}~ (f_1,\ldots,f_k) ~ \Rightarrow~ g, (s, (p, \ch))) \rangle \relderive (s',\ChannelOutput)}
\; p=(a_1,\ldots,a_k)\\
\\
\hline
\\
\inference{\langle p_j, (s, (p, \ch)) \rangle \relderive (s',\ChannelOutput)}{\langle m = \textbf{do}~ p_1 \Box \ldots \Box p_n~ \textbf{od}, (s, (p, \ch))\rangle \relderive (s',\ChannelOutput)}\;j \in \{1,\ldots,n\}\\
\\
\hline
\end{array}
\]
\caption{\label{fig:amdl-semantics-2}Derivation rules for atomic actions, guarded commands, blocks, and middleboxes.}
\end{figure*}

\subsection{Packet State Semantics}
\label{sec:amdl-semanics-packet-space}

The packet state semantics is defined via the evaluation functions
\[
\begin{array}{lcl}
\packeval{\cdot} : \gV{\atom}      & \rightarrow & P \rightarrow (T \cup H)\\
\packeval{\cdot} : \gV{\cond}      & \rightarrow & (\PState[\mboxe] \times P) \rightarrow \{\True,\False\}
\end{array}
\]
and the transition relations
\[
\begin{array}{c}
\packeval{\cdot} : \gV{\action} \rightarrow (\PState[\mboxe] \times P \times \Cont) \times (\PState[\mboxe] \times P \times \Cont)\\
\packeval{ \cdot } : \gV{\gc}  \rightarrow (\PState[\mboxe] \times P \times \Cont) \times (\PState[\mboxe] \times P \times \Cont)\\
\packeval{ \cdot } : \gV{\pblock} \rightarrow (\PState[\mboxe] \times (P \times C_m)) \times (\PState[\mboxe] \times \Cont)\\
\packeval{ \cdot } : \gV{\textit{mbox}} \rightarrow (\PState[\mboxe] \times (P \times C_m)) \times (\PState[\mboxe] \times \Cont) \enspace.\\
\end{array}
\]


\newcommand{\update}{\textit{update}}
We define the helper function
\[
\update : (\PState[m] \times \rels(m) \times \atoms^{*} \times \{\True,\False\}) \rightarrow \PState[m] \enspace,
\]
which updates a given packet state by adding or removing a given tuple from a given relation,
depending on the Boolean value $b$.
\[
\begin{array}{l}
\update(s, r, \overline{a}, b) \eqdef\\
\lambda \tilde{p} \in \packets.\ \lambda q \in Q(m).\
  \begin{cases}
    b, & \text{if $\rel(q)=$ r $\land$}\\
    &\atoms(q)(\tilde{p}) = \overline{a}(p).\\
    \tilde{s}(\tilde{p})(q), & \text{otherwise}.
  \end{cases}
\end{array}
\]

\figref{amdl-semantics-3} shows the evaluation of queries and
the derivation rules for updating relations. The rest of the evaluation functions
and derivation rules have the same shape as those in \figref{amdl-semantics-1} and \figref{amdl-semantics-2},
replacing $\relderive$ with $\packderive$ and $\releval{\cdot}$
with $\packeval{\cdot}$.

\begin{figure*}
\[
\begin{array}{|c|}
\hline
\\
\packeval{\overline{a}~ \bin ~ r} (s,p) \eqdef \left\{
                                                     \begin{array}{ll}
                                                       \False, & s=\err\hbox{;} \\
                                                       s(p)(\overline{a}~ \bin ~ r), & \hbox{otherwise.}
                                                     \end{array}
                                                   \right.\\
\\
\hline
\\
\begin{array}{rll}
\langle r(\overline{a})~ \textbf{:=} ~c, (s,p,\ChannelOutput) \rangle
& \packderive (s,p,\ChannelOutput) & s=\err\\
\langle r(\overline{a})~ \textbf{:=} ~c, (s,p,\ChannelOutput) \rangle
& \packderive (\update(s, r, \overline{a}, b),p,\ChannelOutput) & b=\packeval{c} (s,p)
\end{array}\\
\\
\hline
\end{array}
\]
\caption{\label{fig:amdl-semantics-3}Query evaluation
 and relation update derivation rule for the packet state semantics.}
\end{figure*}

\subsection{Proving \lemref{bisimulation}}
\label{Se:bisimulation_proof}

To prove bisimulation, we use induction on the derivation trees.
Since the shape of all rules, except the ones shown in \figref{amdl-semantics-3},
is exactly the same, we only need to demonstrate bisimilarity for them.

Notice that the semantics is strict in $\err$---the derivation rules for $\err$
propagate $\err$ and query evaluations return \False. We therefore, focus only
on the cases where the states are different from $\err$.

\subsubsection{Bisimilarity of Query Evaluation}
\begin{lemma}
If $\tilde{s} \sim_m s$ and $s\neq\err$ then the following holds:
\[
\packeval{\overline{a}~ \bin ~ r} (\tilde{s},p) = \releval{\overline{a}~ \bin ~ r} (s,p) \enspace.
\]
\end{lemma}

\begin{proof}
Recall that $\tilde{s} \sim_m s$ is defined as:
\[
\tilde{s} = \lambda \tilde{p} \in \packets.\ \lambda q \in Q(m). \
\atoms(q)(\tilde{p}) \in s(\rel(q)) \enspace.
\]

Assume $p=(a_1,\ldots,a_k)$ and $\overline{a}=(f_1,\ldots,f_k)$.

Then the following holds:
\[
\begin{array}{l}
\packeval{\overline{a}~ \bin ~ r} (\tilde{s},p) \\
= \tilde{s}(p)(\overline{a}~ \bin ~ r)\\
= (\lambda \tilde{p} \in \packets.\ \lambda q \in Q(m). \
\atoms(q)(\tilde{p}) \in s(\rel(q))) (p)(\overline{a}~ \bin ~ r)\\
= (\lambda q \in Q(m). \ \atoms(q)(p) \in s(\rel(q))) (\overline{a}~ \bin ~ r)\\
= (a_1,\ldots,a_k) \in s(r)\\
= \releval{(f_1,\ldots,f_k)} p \in s(r)\\
= \releval{\overline{a}~ \bin ~ r} (s,p) \enspace.
\end{array}
\]
\end{proof}

\subsubsection{Bisimilarity of Relation Updates}

\begin{figure*}[t]
\[
\begin{array}{l}
\ps(s[r\mapsto s(r) \cup \{\releval{\overline{a}} p\}])\\
= \lambda \tilde{p} \in \packets.\ \lambda q \in Q(m). \
\atoms(q)(\tilde{p}) \in s[r\mapsto s(r) \cup \{\releval{\overline{a}} p\}](\rel(q)) \\
= \lambda \tilde{p} \in \packets.\ \lambda q \in Q(m). \
\atoms(q)(\tilde{p}) \in \left\{
                           \begin{array}{ll}
                             \{ \releval{\overline{a}} p \}, & \rel(q)=r\hbox{;} \\
                             s(\rel(q)), & \hbox{otherwise.}
                           \end{array}
                         \right.\\
= \lambda \tilde{p} \in \packets.\ \lambda q \in Q(m). \
 \left\{
       \begin{array}{ll}
       \atoms(q)(\tilde{p}) \in \{\releval{\overline{a}} p\}, & \rel(q)=r\hbox{;} \\
       \atoms(q)(\tilde{p}) \in s(\rel(q)), & \hbox{otherwise.}
       \end{array}
 \right.\\
= \lambda \tilde{p} \in \packets.\ \lambda q \in Q(m). \
 \left\{
       \begin{array}{ll}
       \atoms(q)(\tilde{p}) = \releval{\overline{a}} p, & \rel(q)=r\hbox{;} \\
       \atoms(q)(\tilde{p}) \in s(\rel(q)), & \hbox{otherwise.}
       \end{array}
 \right.\\
= \lambda \tilde{p} \in \packets.\ \lambda q \in Q(m). \
 \left\{
       \begin{array}{ll}
       \True, & \rel(q)=r \land \atoms(q)(\tilde{p}) = \overline{a}(p)\hbox{;} \\
       \atoms(q)(\tilde{p}) \in s(\rel(q)), & \hbox{otherwise.}
       \end{array}
 \right.\\
\\
= \lambda \tilde{p} \in \packets.\ \lambda q \in Q(m). \
 \left\{
       \begin{array}{ll}
       \True, & \rel(q)=r \land \atoms(q)(\tilde{p}) = \overline{a}(p)\hbox{;} \\
       \tilde{s}(\tilde{p})(q), & \hbox{otherwise.} \;\;\;\;(\text{using }\ref{Eq:tildef})
       \end{array}
 \right.\\
= \update(\tilde{s}, r, \overline{a}, \True)\\
\end{array}
\]
\caption{\label{fig:subproofAdd}Detailed proof steps.}
\end{figure*}

Assume that $\tilde{s} \sim_m s$ and that $s\neq\err$.
By the induction hypothesis, we have that $b=\packeval{c} (\tilde{s},p)=\releval{c} (s,p)$ holds.

Assume that $b=\True$. Therefore, the following derivations apply:
\[
\begin{array}{l}
\langle r(\overline{a})~ \textbf{:=} ~c, (s,p,\ChannelOutput) \rangle \relderive
(s[r\mapsto s(r) \cup \{\releval{\overline{a}} p\}],p,\ChannelOutput) \\
\langle r(\overline{a})~ \textbf{:=} ~c, (\tilde{s},p,\ChannelOutput) \rangle \packderive
(\update(\tilde{s}, r, \overline{a}, \True),p,\ChannelOutput)\enspace.
\end{array}
\]

We will use the following identity, which we obtain from the definition of $\tilde{s}$:
\begin{equation}\label{Eq:tildef}
\begin{array}{l}
\tilde{s}(p)(q)\\
= (\lambda \tilde{p} \in \packets.\ \lambda \tilde{q} \in Q(m). \
\atoms(\tilde{q})(\tilde{p}) \in s(\rel(\tilde{q}))) (p) (q) \\
= \atoms(q)(p) \in s(\rel(q)) \enspace.
\end{array}
\end{equation}

We have to show that the following relation holds in \figref{subproofAdd}:
\[
s[r\mapsto s(r) \cup \{\releval{\overline{a}} p\}] \sim_m \update(\tilde{s}, r, \overline{a}, \True) \enspace.
\]

Assume that $b=\False$. Therefore, the following derivations apply:
\[
\begin{array}{l}
\langle r(\overline{a})~ \textbf{:=} ~c, (s,p,\ChannelOutput) \rangle \relderive
(s[r\mapsto s(r) \setminus \{\releval{\overline{a}} p\}],p,\ChannelOutput) \\
\langle r(\overline{a})~ \textbf{:=} ~c, (\tilde{s},p,\ChannelOutput) \rangle \packderive
(\update(\tilde{s}, r, \overline{a}, \False),p,\ChannelOutput)\enspace.
\end{array}
\]

We show that the following relation holds in \figref{subproofRemove}:
\[
s[r\mapsto s(r) \setminus \{\releval{\overline{a}} p\}] \sim_m \update(\tilde{s}, r, \overline{a}, \False) \enspace.
\]

\begin{figure*}[t]
\[
\begin{array}{l}
\ps(s[r\mapsto s(r) \setminus \{\releval{\overline{a}} p\}])\\
= \lambda \tilde{p} \in \packets.\ \lambda q \in Q(m). \
\atoms(q)(\tilde{p}) \in s[r\mapsto s(r) \setminus \{\releval{\overline{a}} p\}](\rel(q)) \\
= \lambda \tilde{p} \in \packets.\ \lambda q \in Q(m). \
 \left\{
       \begin{array}{ll}
       \atoms(q)(\tilde{p}) \not\in \{\releval{\overline{a}} p\}, & \rel(q)=r\hbox{;} \\
       \atoms(q)(\tilde{p}) \in s(\rel(q)), & \hbox{otherwise.}
       \end{array}
 \right.\\
= \lambda \tilde{p} \in \packets.\ \lambda q \in Q(m). \
 \left\{
       \begin{array}{ll}
       \atoms(q)(\tilde{p}) \neq \releval{\overline{a}} p, & \rel(q)=r\hbox{;} \\
       \atoms(q)(\tilde{p}) \in s(\rel(q)), & \hbox{otherwise.}
       \end{array}
 \right.\\
= \lambda \tilde{p} \in \packets.\ \lambda q \in Q(m). \
 \left\{
       \begin{array}{ll}
       \False, & \rel(q)=r \land \atoms(q)(\tilde{p}) = \overline{a}(p)\hbox{;} \\
       \atoms(q)(\tilde{p}) \in s(\rel(q)), & \hbox{otherwise.}
       \end{array}
 \right.\\
\\
= \lambda \tilde{p} \in \packets.\ \lambda q \in Q(m). \
 \left\{
       \begin{array}{ll}
       \False, & \rel(q)=r \land \atoms(q)(\tilde{p}) = \overline{a}(p)\hbox{;} \\
       \tilde{s}(\tilde{p})(q), & \hbox{otherwise.} \;\;\;\;(\text{using }\ref{Eq:tildef})
       \end{array}
 \right.\\
= \update(\tilde{s}, r, \overline{a}, \False)\\
\end{array}
\]
\caption{\label{fig:subproofRemove}Detailed proof steps.}
\end{figure*} 
\section{Hierarchy of Abstract Domains} 
\label{sec:figures}

\figref{CombinedAbstractions} provides a high-level view of the different
network semantics.

\begin{figure}[t]
  \centering
  \begin{tabular}{|c|}
  \hline
  Cartesian network domain over Cartesian packet space domain:\\
  $M\rightarrow (P\rightarrow \pow{Q\rightarrow \{\True,\False\}}) \times (E\rightarrow \pow{\Pi})$\\
  \hline\\
  \includegraphics[width=0.99\textwidth]{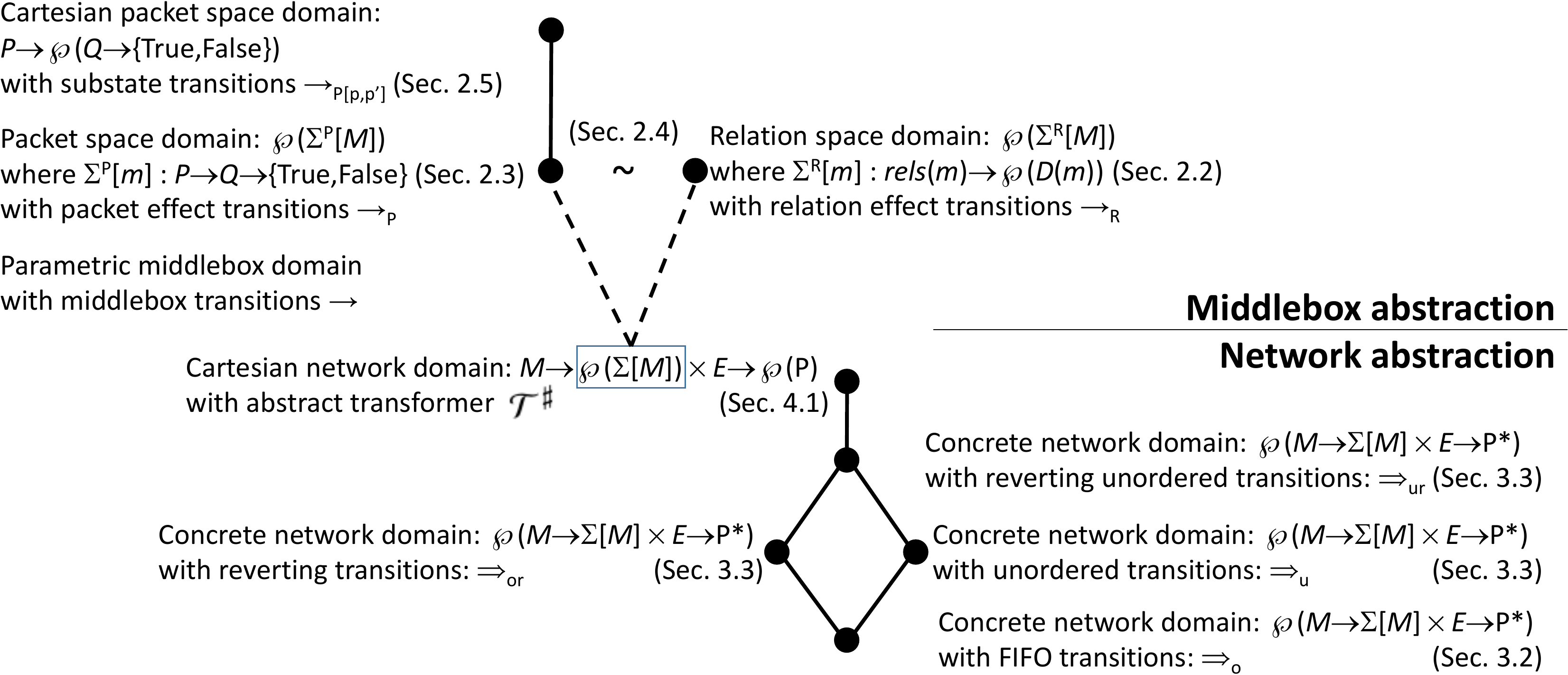}\\
  \hline
  \end{tabular}
  \caption{Hierarchy of abstractions.
  Solid edges stand for abstraction (either by relaxing the
  transition relation or by abstracting the configurations).
  Dashed edges stand for instantiation of the middlebox
  (local) semantics.}
  \label{fig:CombinedAbstractions}
\end{figure}


\section{Example} 
\label{sec:examples}

\begin{figure}[t]
\centering
\begin{subfigure}[b]{0.8\textwidth}
\centering
\includegraphics[width=\textwidth]{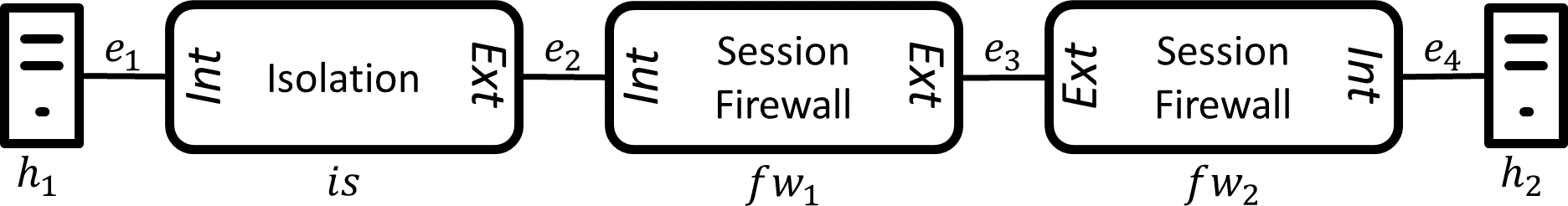}
\caption{A network topology.}

\label{fig:running-ex-topo}
\end{subfigure}
\begin{subfigure}[b]{0.5\textwidth}
\begin{footnotesize}
\begin{alltt}
\begin{tabbing}
is\= = do \+ \\
ex\=ternal_port ? p =>\+\\
if\= \+ \\
p.src = forbidden => abort\\
\(\Box\)\\
true => internal_port ! p \-\\
fi \- \\
\(\Box\)\\
in\=ternal_port ? p => \+ \\
true => external_port ! p \-\- \\
od
\end{tabbing}
\end{alltt}
\end{footnotesize}
\caption{ \label{fig:safety-code}%
\AMDL\ code for $\is$.}
 \end{subfigure}
\caption{Network topology and \AMDL\ code for the running example.}
 \end{figure}

\figref{running-ex-topo} shows a simple network where two stateful
firewalls are connected in a row to prevent traffic  between nodes $h_2$ to
$h_1$. This is an artificial example meant to illustrate the verification
process. More realistic examples are presented in \secref{empirical}. It is
assumed that hosts $h_1$ and $h_2$ can send and receive arbitrary packets on
channels $e_1$ and $e_4$, respectively. The example is implemented using three
middleboxes: two middleboxes, $\fw_1$ and $\fw_2$, running firewalls that
restrict traffic from left to right and from right to left, respectively, and
one middlebox, $\is$, checking whether isolation between $h_2$ and $h_1$ is
preserved. In $\fw_1$, $e_2$ is connected to the ``internal'' port and $e_3$ is
connected to the ``external'' port, thus limiting traffic from right to left. In
$\fw_2$, $e_4$ is connected to the ``internal'' port and $e_3$  is connected to
the ``external'' port, thus limiting traffic from left to right. In $\is$, $e_1$
is connected to the ``internal'' port and $e_2$ is connected to the ``external''
port.

\figref{session-fw-code} describes the code running in either of the session
firewalls, $\fw_1$ and $\fw_2$. We use CSP/OCCAM-like syntax where (messages)
packets are sent/received asynchronously. The middlebox non-deterministically
operates on a packet from the ``internal'' port  or the ``external'' port. When
reading a packet from the ``internal'' port, the program distinguishes between
two cases. In the first case, a session had been previously established, and the
packet is simply forwarded to the ``external'' port.
In the second case the type of the packet is a ``request'' packet (\type=0), and
the program adds the destination host to the set of \verb|requested| hosts and
forwards the ``request''. The \verb|requested| set is used to store the hosts to
which the middlebox sent a ``request'' packet, to avoid the case where a session
is established with a host that the middlebox did not send a ``request'' to.
Packets that do not fall into any of these two cases are discarded with no
further processing.

When the middlebox reads a packet from the ``external'' port again it
distinguishes between two cases --- in one case a session had previously been
established, and is similar to its ``internal'' counterpart. In the second case,
the processed packet is a ``response'' packet (\type=1) from a host that is in
the \verb|requested| set, and the program marks the source of the packet as
\verb|trusted|, thus establishing a session. Other packets are discarded.

A ``data'' packet (\type=2) is implicitly handled by checking whether the
source/destination of the packet is in the \texttt{trusted} set, and if so,
allowing the packet to propagate on.

\figref{safety-code} describes the code running in a special middlebox, $\is$,
which intercepts packets before they arrive to host $h_1$ --- the middlebox non-
deterministically reads a packet from the ``external'' port and aborts if the
source of the packet is the host $\texttt{forbidden} = h_2$, and otherwise
forwards to $h_1$ on the ``internal'' port. On the other direction, it simply
forwards packets from the ``internal'' port to the ``external'' port. In this
example, $\is$ models the safety property.

\begin{table*}[t]
\begin{footnotesize}
\[\arraycolsep=2pt
\begin{array}{|c | c | c | c |  c | c | c | c| c | c | l|}
\hline \mathbf{\overset\rightarrow{e_1}} & \mathbf{\overset\leftarrow{e_1}}
& \mathbf{\overset\rightarrow{e_2}} & \mathbf{\overset\leftarrow{e_2}} &
\mathbf{\fw_1} &  \mathbf{\overset\rightarrow{e_3}} &
\mathbf{\overset\leftarrow{e_3}} & \mathbf{\fw_2}  &
\mathbf{\overset\rightarrow{e_4}} & \mathbf{\overset\leftarrow{e_4}} & \textbf{action} \\
\hline \hline
\begin{array}{l}
p_{(1,2,0)}\\
p_{(1,2,1)}\\
p_{(1,2,2)}\end{array} & \emptyset & \emptyset & \emptyset &
 (\emptyset, \emptyset) &
\emptyset & \emptyset & (\emptyset, \emptyset) & \emptyset
&
\begin{array}{l}
p_{(2,1,0)}\\
p_{(2,1,1)}\\
p_{(2,1,2)}\end{array}
 & \text{initial state} \\
\hline
& & p_{(1,2,0)} & & & & & & & & \textit{is} ~\text{reads}~ p_{(1,2,0)}\\
\hline
& &
\begin{array}{l}
p_{(1,2,0)}\\
p_{(1,2,1)}\end{array}
& & & & & & & & \textit{is} ~\text{reads}~ p_{(1,2,1)}\\
\hline
& &
\begin{array}{l}
p_{(1,2,0)}\\
p_{(1,2,1)}\\
p_{(1,2,2)}\end{array}
& & & & & & & & \textit{is} ~\text{reads}~ p_{(1,2,2)}\\
\hline
&  &  &  & \begin{array}{l}
(\emptyset, \emptyset)\\
(\{h_2\}, \emptyset)
\end{array}
& p_{(1,2,0)} & & & & & \textit{fw}_1~ \text{reads}~ p_{(1,2,0)}\\

\hline

& & & & & & & & & & \textit{fw}_1~ \text{reads}~ p_{(1,2,1)}\\

\hline

& & & & & & & & & & \textit{fw}_1~ \text{reads}~ p_{(1,2,2)}\\

\hline

& & & & & & p_{(2,1,0)} &
\begin{array}{l}
(\emptyset, \emptyset)\\
(\{h_1\}, \emptyset)
\end{array}
& & & \textit{fw}_2~ \text{reads}~ p_{(2,1,0)}\\

\hline

& & & & & & & & & & \textit{fw}_2~ \text{reads}~ p_{(2,1,1)}\\

\hline

& & & & & & & & & & \textit{fw}_2~ \text{reads}~ p_{(2,1,2)}\\

\hline

& & & & & & & & & & \textit{fw}_1~ \text{reads}~ p_{(2,1,0)}\\

\hline

& & & & & & & & & & \textit{fw}_2~ \text{reads}~ p_{(1,2,0)}\\

\hline
\hline
\end{array}
\]
\end{footnotesize}
\caption{\label{tab:CartesianExplicit}%
Modular analysis of the running example with explicit state representation. Only
changed values are shown. The abstract states of channels are sets. The abstract
states of firewalls are sets of pairs for the values of the \texttt{requested}
and \texttt{trusted} sets. Each cell in the table represent a set of the
elements described within, except for empty sets in the initial state. The
notation $p_{(i,j,k)}$ stands for the packet from $h_i$ to $h_j$ with type $k$.
}
\end{table*}

\subsection{Analysis Using Network Level Abstractions}
\tabref{CartesianExplicit} shows the
run of our analysis, when restricted to the network-level abstractions, on the
running example. Each row corresponds to a step in the least fixpoint
computation of the (abstract) reachable network states. Each column at the table
represents the abstract content of a channel (as a set of packets) or the
abstract state of an individual middlebox (as the contents of its set-valued
variables). For each channel $e$, $\overset\rightarrow{e}$ denotes channels
connecting traffic from left to the right, while $\overset\leftarrow{e}$ denotes
channels connecting traffic from right to the left. For example,
$\overset\rightarrow{e_1}$ contains packets sent from $h_1$ to $\is$.

Channel abstract states are sets of packets.

For the firewall middleboxes, a (concrete) state is a pair of values for the
\verb|requested| and \verb|trusted| sets. An abstract state is a set of such
(concrete) states.
The isolation middlebox is stateless.

At the initial configuration, the states of $\fw_1$ and $\fw_2$ are pairs of
empty sets; the states of channels $\overset\rightarrow{e_1}$ and
$\overset\leftarrow{e_4}$ are all the packets that hosts $h_1$ and $h_2$ can
send, respectively.

The analysis ignores the correlations between different columns. At each step,
the analysis chooses an input channel and a middlebox state and computes the
next state. The analysis stops when no more new middlebox states or channel
states are discovered and reports potential violation of the safety property if
the \texttt{abort} command is executed.

In the first action, the code of $\is$ executes and reads $(h_1, h_2, 0)$ from
$\overset\rightarrow{e_1}$. Notice that this does not change the (abstract)
content of this channel. The packet is forwarded to $\overset\rightarrow{e_2}$.
Thus, our analysis only accumulates packets, ignoring their order. The reachable
states of the middleboxes are explicitly maintained. For example, when $\fw_1$
reads $(h_1, h_2, 0)$ from $\overset\rightarrow{e_2}$, it forwards it to
$\overset\rightarrow{e_3}$ and reaches a new state with
$\texttt{requested}=\{h_2\}$ and $\texttt{trusted}=\emptyset$.

Notice that in this example, the analysis proved that the \texttt{abort} command
can ever be executed on arbitrary packet propagation scenarios. Specifically, no
packets ever reaches channel $\overset\leftarrow{e_2}$, so the safety middlebox
$\is$ never reads a packet that will result in the execution of an
\texttt{abort} command. Thus, the analysis succeeded in proving isolation.

This example illustrates that, although our analysis employs Cartesian
abstraction, it is able to prove a network-wide property. Specifically, proving
isolation requires reasoning about the states of both firewalls. We note that
removing either of the firewalls violates the safety property.

\begin{table*}
\begin{footnotesize}
\[\arraycolsep=2pt
\begin{array} {|c | c | c | c | c | c | c | c | c | c | l|}
\hline \mathbf{\overset\rightarrow{c_1}} & \mathbf{\overset\leftarrow{c_1}} &
\mathbf{\overset\rightarrow{c_2}} & \mathbf{\overset\leftarrow{c_2}} &
\fw_1 & \mathbf{\overset\rightarrow{c_3}} &
\mathbf{\overset\leftarrow{c_3}} & \fw_2 &
\mathbf{\overset\rightarrow{c_4}}
 & \mathbf{\overset\leftarrow{c_4}} & \textbf{action} \\

\hline \hline

\begin{array}{l}
p_{(1,2,0)}\\
p_{(1,2,1)}\\
p_{(1,2,2)}\end{array}  & \emptyset & \emptyset & \emptyset &

\begin{array}{l}
p_{(1,2,0)} \mapsto (F,F,F) \\
p_{(1,2,1)} \mapsto (F,F,F) \\
p_{(1,2,2)} \mapsto (F,F,F) \\
p_{(2,1,0)} \mapsto (F,F,F) \\
p_{(2,1,1)} \mapsto (F,F,F) \\
p_{(2,1,2)} \mapsto (F,F,F) \\
\end{array}  & \emptyset & \emptyset &

\begin{array}{l}
p_{(1,2,0)} \mapsto (F,F,F) \\
p_{(1,2,1)} \mapsto (F,F,F) \\
p_{(1,2,2)} \mapsto (F,F,F) \\
p_{(2,1,0)} \mapsto (F,F,F) \\
p_{(2,1,1)} \mapsto (F,F,F) \\
p_{(2,1,2)} \mapsto (F,F,F) \\
\end{array}  & \emptyset &

\begin{array}{l}
p_{(2,1,0)}\\
p_{(2,1,1)}\\
p_{(2,1,2)}\end{array}  & \text{initial state} \\
\hline
& & p_{(1,2,0)} & & & & & & & & \textit{is} ~\text{reads}~ p_{(1,2,0)}\\
\hline
& &

\begin{array}{l}
p_{(1,2,0)}\\
p_{(1,2,1)}\end{array}
& & & & & & & & \textit{is} ~\text{reads}~ p_{(1,2,1)}\\
\hline
& &

\begin{array}{l}
p_{(1,2,0)}\\
p_{(1,2,1)}\\
p_{(1,2,2)}\end{array}
& & & & & & & & \textit{is} ~\text{reads}~ p_{(1,2,2)}\\

\hline

 &  &  &  &

\begin{array}{l}
p_{(1,2,0)} \mapsto (F,F,F) \\
p_{(1,2,1)} \mapsto (F,F,F) \\
p_{(1,2,2)} \mapsto (F,F,F) \\
p_{(2,1,0)} \mapsto (F,F,F) \\
p_{(2,1,1)} \mapsto (F,F,F) \\
p_{(2,1,2)} \mapsto (F,F,F) \\
\hline
p_{(2,1,0)} \mapsto (F,F,T) \\
p_{(2,1,1)} \mapsto (F,F,T) \\
p_{(2,1,2)} \mapsto (F,F,T) \\
\end{array}  & p_{(1,2,0)}  &  &  &  & & \textit{fw}_1~\text{reads}~ p_{(1,2,0)} \\

\hline

& & & & & & & & & & \textit{fw}_1~ \text{reads}~ p_{(1,2,1)}\\

\hline

& & & & & & & & & & \textit{fw}_1~ \text{reads}~ p_{(1,2,2)}\\

\hline

 &  &  & & & &p_{(2,1,0)}  &

\begin{array}{l}
p_{(1,2,0)} \mapsto (F,F,F) \\
p_{(1,2,1)} \mapsto (F,F,F) \\
p_{(1,2,2)} \mapsto (F,F,F) \\
p_{(2,1,0)} \mapsto (F,F,F) \\
p_{(2,1,1)} \mapsto (F,F,F) \\
p_{(2,1,2)} \mapsto (F,F,F) \\
\hline
p_{(1,2,0)} \mapsto (F,F,T) \\
p_{(1,2,1)} \mapsto (F,F,T) \\
p_{(1,2,2)} \mapsto (F,F,T) \\
\end{array}   &  &  & \textit{fw}_2~\text{reads}~ p_{(2,1,0)} \\

\hline

& & & & & & & & & & \textit{fw}_2~ \text{reads}~ p_{(2,1,1)}\\

\hline

& & & & & & & & & & \textit{fw}_2~ \text{reads}~ p_{(2,1,2)}\\

\hline

& & & & & & & & & & \textit{fw}_1~ \text{reads}~ p_{(2,1,0)}\\

\hline

& & & & & & & & & & \textit{fw}_2~ \text{reads}~ p_{(1,2,0)}\\

\hline
\hline

\end{array}
\]
\caption{\label{tab:PacketSpaceEnumeration}%
Packet state enumeration for the running example.
The abstract states of channels are sets of packets.
The abstract states of middleboxes are relations over packets and query valuations; each entry in the table is denoted by $\mapsto$.
Each cell in the table represent a set of the elements described within, except for empty sets in the initial state.
The horizontal lines in the $\fw_1$ and $\fw_2$ columns appear to emphasize the changes.
As before, $p_{(i,j,k)}$ stands for the packet from $h_i$ to $h_j$ with type $k$.}
\end{footnotesize}
\end{table*}

\subsection{Analysis Using Network Level and Middlebox Level Abstractions}
\tabref{PacketSpaceEnumeration} shows the verification process with packet
states in the running example. Instead of storing the contents of relations
\verb|trusted| and \verb|requested| in each middlebox state, we store, for each
packet, whether each of the expressions ``\verb|p.dst in trusted|'',
``\verb|p.src in trusted|'', and ``\verb|p.src in requested|'', evaluates to
\True\ ($T$) or \False\ ($F$), respectively.

Since both relations are empty in the initial state, the packet states for both
firewalls map each packet to $(F,F,F)$.

Recall that when $\fw_1$ reads $(h_1, h_2, 0)$ from $\overset\rightarrow{e_2}$,
it forwards it to $\overset\rightarrow{e_3}$ and reaches a new state with
$\requested=\{h_2\}$ and $\trusted=\emptyset$.
Therefore, any future evaluation of the expression ``\verb|p.src in requested|''
(for any value of \type) should result in \True.
Under the packet state representation, this would result in adding to the
abstract state of $\fw_1$ a packet state similar to that of the initial state
where each of the packets $p_{(2,1,0)}$, $p_{(2,1,1)}$, and $p_{(2,1,2)}$ is re-
mapped from $(F,F,F)$ to $(F,F,T)$.
Our middlebox-level Cartesian abstraction allows us to instead accumulate these
mappings (separated by a horizontal line from the initial mappings) in a single
abstract state, \emph{without affecting the overall precision of the abstract
interpretation}.

A similar change to the packet state of $\fw_2$ occurs upon reading the packet
$(h_2, h_1, 0)$ from $\overset\leftarrow{e_4}$.




\section{Networks with unbounded number of hosts}\label{sec:small-model}
In this section, we prove the lack of small model to stateful networks, w.r.t
number of network hosts. This property holds even for reverting networks with
only a single middlebox and packets of the type $(s,d,t)$ where $s$ and $d$ are
hosts i.e., $s,d\in H$, and $t$, the packet type, is taken from a bounded type
set $T$.

\subsubsection*{Small model property.}
For simplicity, we consider only a network with a single middlebox $m$ that
never output packets. The small model property is a bound $b(m)$, such that any
network with the above topology is safe if and only if any network with the
above topology and at most $b(m)$ hosts is safe. And if for certain number of
hosts the network is not safe, we define $b(m) = \infty$.

\begin{theorem}
The function $b(m)$ is not a computable function. In particular, the problem of
deciding whether $b(m) < \infty$ is undecidable.
\end{theorem} We prove the above theorem by a reduction to the halting problem.
We show that giving a Turing machine $M$, we can construct a middlebox $m(M)$
such that $b(m(M)) = \infty$ if and only if $M$ is never halts and is using
unbounded space on its run when then initial input is empty (which is known to
be undecidable).

\subsubsection*{Proof overview}
Given a Turing machine $M$ over alphabet $\sigma$ we construct a network with a
single middlebox $m$ and a host set $H$ and packet space $P=H\times H\times T$
such that $N$ is safe if and only if $M$ does not halts for any run that
requires at least $|H|$ space.

Informally, we construct $m$ such that initially $m$ encodes a successor
relation over $H$, and later it uses the relation to simulate the run of the
Turing machine $m$ for $|H|$ cells in the turing machine tape. If in using at
most $|H|$ space the Turing machine halts, then $m$ goes to an abort state.
Hence, $N$ is safe iff $M$ does not halt using at most $|H|$ space.

\subsubsection*{Detailed proof sketch}
We assume a constant symbol $h_0$ (the first host). For the successor
construction, the middlebox $m$ has the next relations:
\begin{itemize}
\item $R_{\mbox{successor}}(h_1,h_2)$. Intuitively, $R_{successor}(h_1,h_2) =
True$ stands for $h_1 = h_2 + 1$. Initially, the relation returns false to all
pairs.
\item $R_{\mbox{max host}}(h)$. Intuitively, $R_{\mbox{max host}}(h) = True$, if
$h$ was the last host that was assigned as a successor. Initially, only
$R_{\mbox{max host}}(h_0) = True$.
\item $R_{\mbox{already in order}}(h)$. Intuitively, $R_{\mbox{already in
order}}(h) = True$ if $h$ was already assigned as a successor. Initially only
$R_{\mbox{already in order}}(h_0) = True$.
\end{itemize}

In the successor construction phase, $m$ construct an order, given an input
packet $(s,d,t)$ as follows: If $R_{\mbox{max host}}(s)$ is false or
$R_{\mbox{already in order}}(d)$ is true, it goes to a sink state. Otherwise it
set $R_{\mbox{max host}}(s) = False$, $R_{\mbox{already in order}}(d) = True$,
$R_{\mbox{max host}}(d) = True$ and $R_{successor}(s,d) = True$. A special
packet type $t = 1$ indicates that $m$ should leave the successor construction
phase and go to simulation phase.

To describe the simulation phase, we first recall that a Turing machine has a
finite set of states $Q$ and a finite input/output alphabet $\Sigma$. In every
step, the machine reads an input from the head, write a new symbol to head, and
moves the head one step to the right or to the left (w.l.o.g, we assume that
head position is changed in every step). At this phase, hosts represent turing
machine head position. For the Turing machine simulation phase the middlebox has
the next relations:

\begin{itemize}
\item For every $\sigma\in\Sigma$: $R_{\mbox{symbol}_\sigma}(h)$. Intuitively,
it is true if and only if the symbol on the $h-th$ position is $\sigma$.
Initially, it is false for all pairs.
\item $R_{\mbox{expected position}}(h)$. Intuitively, it is true if and only if
the head is expected to be in position $h$. Initially, only $R_{\mbox{expected
position}}(h_0)$ is true.
\item For every $q\in Q$: $R_{\mbox{state}_q}()$ is true iff the machine is at
state $q$. Initially, only $R_{\mbox{state}_{q_0}}()$ is true.
\end{itemize}

In this state, $m$ simulates the machine as follows:
given a packet $(s,d,t)$:
\begin{itemize}
\item Check head position: If $R_{\mbox{expected position}}(s) = False$ go to
sink state.
\item Query head symbol: go over all $R_{\mbox{symbol}_\sigma}(s)$ and extract
current head symbol $\sigma$ (if it is false for all symbols, then the cell is
empty, i.e., $\sigma = \epsilon$).
\item Query current state: go over all $R_{\mbox{state}_q}()$ and extract
current state $q$.
\item Update head symbol and current state: set $R_{\mbox{symbol}_\sigma}(s) =
False$ and $R_{\mbox{symbol}_\sigma'}(s) = True$ where $\sigma'$ is the output
symbol (according to the turing machine). Similarly update the current state
relation.
\item Update expected head position: If at state $q$ and input $\sigma$ the head
moves left, then if $d \neq s-1$ (according to the successor relation) then go
to sink state. Otherwise set $R_{\mbox{expected position}}(s) = False$, and
$R_{\mbox{expected position}}(d) = True$. If the head moves right, check if $d =
s+1$ and act in the same way.
\item if $q$ is a final state, then abort.
\end{itemize}

\begin{lemma}
The network is safe if and only if $M$ does not halt using at most $|H|$ space.
\end{lemma}
\begin{proof}
If $M$ halts using at most $|H|$ space, then a sequence of packets which
construct the order and simulate the run without going to a sink state leads to
an abort state. If $M$ does not halt with at most $|H|$ space, then any sequence
of packets must end in a sink state.
\end{proof}

\subsubsection*{Additional observations}
\begin{itemize}
    \item The program is only using the inputs $s,d$ and $t$ and a single
    constant $h_0$. In the construction it is enough to have $t\in\{0,1\}$.
    \item Same proof holds for reverting middlebox. Indeed, whenever the
    middlebox reverts, the state of the turing machine and the relation order
    are reset, and the run starts from scratch. This is thanks to the fact that
    $m$ does not output any packets.
\end{itemize}


\fi

\end{document}